\documentclass[a4paper,UKenglish,cleveref, autoref]{lipics-v2019}

\usepackage{amsmath}
\usepackage{mathtools}
\usepackage{todonotes}
\usepackage[]{algorithm2e}
\usepackage{xcolor}

\newcommand{\lara}{{\sc Lara}}
\newcommand{\foagg}{{\rm FO}_{\sf Agg}}
\newcommand{\matlang}{{\sc Matlang}}
\newcommand{\tp}{\texttt{:}}
\newcommand{\map}{{\sf Map}}
\newcommand{\JoinOp}[3]{#1\bowtie_{#3}#2}
\newcommand{\AggOp}[3]{\union^{#2}_{#3}#1}
\newcommand{\reduction}{\mathrel{\bar{\hspace{0.7pt}\union}}}
\newcommand{\RedOp}[3]{\reduction^{#2}_{#3}#1}
\newcommand{\rest}[2]{{#1}{\downarrow_{#2}}}

\title{On the Expressiveness of \lara: A Unified Language for Linear and Relational Algebra} 

\titlerunning{On the expressive power of \lara}

\author{Pablo Barcel{\'{o}}}{IMC, Pontificia Universidad Cat\'olica de Chile \& IMFD Chile}{pbarcelo@ing.puc.cl}{}{}
\author{Nelson Higuera}{Department of Computer Science, University of Chile \& IMFD Chile}{nhiguera@dcc.uchile.cl}{}{}
\author{Jorge P\'erez}{Department of Computer Science, University of Chile \& IMFD Chile}{jperez@dcc.uchile.cl}{}{}
\author{Bernardo Subercaseaux}{Department of Computer Science, University of Chile \& IMFD Chile}{bsuberca@dcc.uchile.cl}{}{}

\authorrunning{P. Barcel\'o, N. Higuera, J. P\'erez, and B. Subercaseux}



\nolinenumbers 

\EventShortTitle{ICDT 2020}
\EventAcronym{ICDT}
\EventYear{2020}
\ArticleNo{}

\usepackage{tikz}

\makeatletter
\newcommand{\union}{}
\DeclareRobustCommand{\union}{\mathrel{\mathpalette\hour@glass\relax}}

\newcommand\hour@glass[2]{%
  \vcenter{\hbox{%
    \rotatebox[origin=c]{90}{$\m@th#1\bowtie$}%
  }}%
}
\makeatother

\begin{document}

\maketitle

\begin{abstract}
  We study the expressive power of the \lara\ language -- a recently proposed unified model for expressing relational and linear algebra operations -- both in terms of traditional database
  query languages and some analytic tasks often performed in machine learning pipelines.
  We start by showing \lara\ to be expressive complete with respect to first-order logic with aggregation. Since \lara\ is parameterized by a set of user-defined functions which allow to transform values in tables, the exact expressive power of the language depends on how these functions are defined. We distinguish two main cases depending on the level of genericity queries are enforced to satisfy. Under strong genericity assumptions the language cannot express matrix convolution, a very important operation in current machine learning operations. This language is also local, and thus cannot express operations such as matrix inverse that exhibit a recursive behavior. For expressing convolution, one can relax the genericity requirement by adding
  an underlying linear order on the domain. This, however, destroys locality and turns the expressive power of the language much more difficult to understand. In particular, although under complexity assumptions the resulting language can still not express matrix inverse, a proof of this fact without such assumptions seems challenging to obtain.
  \end{abstract}

\section{Introduction} \label{sec:intro}
Many of the actual analytics systems require both relational algebra and statistical functionalities for manipulating the data. In fact, while 
tools based on relational algebra are often used for preparing and structuring the data, the ones based on statistics and machine learning (ML) 
are applied to quantitatively reason about such data. 
Based on the ``impedance mismatch'' that this dichotomy creates \cite{lara-opt}, the database theory community has highlighted the need of developing 
a standard data model and query language for such applications, meaning an extension of relational algebra with linear algebra operators that is able to express
the most common ML tasks \cite{dagstuhl}. Noticeably, the ML community has also recently manifested the need for what --  at least from a database perspective -- 
can be seen as a high-level language that manipulates tensors.
Indeed, despite their wide adoption, there has been a recent interest in redesigning the way in which tensors are used  
in deep learning code~\cite{xarray,TensorHarmful,TensorHarmful2}, due to 
some pitfalls of the current way in which tensors are abstracted.

Hutchinson et al.~\cite{lara,lara1} have recently proposed a data model and a query language that aims at becoming the ``universal 
connector'' that solves the aforementioned impedance.  On the one hand, the data model proposed corresponds to the so-called 
{\em associative tables}, which generalize relational tables, tensors, arrays, and others. Associative tables are two-sorted, consisting of 
{\em keys} and {\em values} that such keys map to. The query language, on the other hand, 
is called \lara, and subsumes several known languages for the data models mentioned above.    
 \lara\ is an algebraic language 
 designed in a minimalistic way by only including three operators; namely, {\em join}, {\em union}, and 
 {\em extension}. In rough terms, the first one corresponds to the traditional join from relational algebra, the second one to the 
 operation of aggregation, and the third one to the extension defined by a function as in a {\em flatmap} operation.  
 It has been shown that \lara\ subsumes all relational algebra operations and 
is capable of expressing several interesting linear algebra operations used in graph algorithms \cite{lara1}. 

Based on the proposal of \lara\ as a unified language for relational and linear algebra, it is relevant 
to develop a deeper understanding of its expressive power, both in terms of the logical query languages traditionally studied in database theory and 
 ML operations often performed in practical applications. We start with the former and show that \lara\ is expressive complete with respect to {\em first-order logic with aggregation} ($\foagg$), a language that has been studied as a way to abstract the expressive power of SQL without recursion; cf., \cite{Lib03,Lib04}. (To be more precise, \lara\ is expressive complete with respect to a suitable syntactic fragment of $\foagg$ that ensures that 
formulas are {\em safe} and get properly evaluated over associative tables). 
This result can be seen as a sanity check for \lara. In fact, this language 
is specifically tailored to handle aggregation in conjunction with relational algebra operations, and a classical result in database theory establishes that 
the latter is expressive complete with respect to first-order logic (FO). \textcolor{blue}{We observe that while \lara\ consists of  
{\em positive} algebraic operators only, set difference can be encoded in the language by a combination of 
aggregate operators and extension functions.}
Our expressive completeness 
result is parameterized by the class of functions allowed 
in the extension operator. For each such a class $\Omega$ we allow $\foagg$ to contain all built-in predicates that encode the functions in $\Omega$.

To understand which ML operators \lara\ can express, one then needs to bound the class $\Omega$ of extension functions allowed in the language. 
We start with a tame class that can still express several relevant functions. These are the FO-expressible 
functions that allow to compute arbitrary numerical predicates on values, but can only compare keys with respect to (in)equality. 
This restriction makes the logic quite amenable for theoretical exploration. In fact, it is easy to show that 
the resulting ``tame version'' of \lara\ satisfies a strong {\em genericity} criterion (in terms of key-permutations) and is also {\em local}, in the sense 
that queries in the language can only see up to a fixed-radius neighborhood from its free variables; cf., \cite{Lib04}. 
The first property implies that this tame version of \lara\ cannot express non-generic operations, such as matrix {\em convolution}, and 
the second one that it cannot express inherently recursive queries, such as matrix {\em inverse}. Both operations are very relevant for
ML applications; e.g., matrix convolution is routinely applied in dimension-reduction tasks, 
while matrix inverse is used for learning the matrix of coefficient values in linear regression. 

We then look more carefully at the case of matrix convolution, and show that this query 
can be expressed if we relax the genericity properties of the language 
by assuming the presence of a linear order on the domain of keys. (This relaxation implies that queries 
expressible in the resulting version of \lara\ are no longer invariant with respect to key-permutations). This language, however, 
is much harder to understand in terms of its expressive power. In particular, it can express non-local queries, and hence we cannot apply locality 
techniques to show that the matrix inversion query is not expressible in it.  
To prove this result, then, one would have to apply techniques based on the {\em Ehrenfeucht-Fra\"iss\'e} games that characterize the expressive power of 
the logic. Showing results based on such games in the presence of a linear order, however, is often combinatorially difficult, and currently we do not know 
whether this is possible. In turn, it is possible to obtain that matrix inversion is not expressible in a natural restriction of our language 
under complexity-theoretic assumptions. This is because the data complexity of queries expressible in 
such a restricted language  is {\sc Logspace}, while matrix inversion is complete for a class that is believed to be a proper extension of 
the latter.  

The main objective of our paper is connecting the study of the expressive power of tensor-based query languages, in general, and 
of \lara, in particular, with traditional database theory concepts and the arsenal of techniques that have been developed in this area 
to study the expressiveness of query languages. We also aim at identifying potential lines for future research that appear 
in connection with this problem. 
Our work is close in spirit 
to the recent study of {\sc Matlang} \cite{BGBW18,Geerts19}, a matrix-manipulation language based on 
elementary linear algebra operations. It is shown that this language is contained in the three-variable fragment of relational algebra 
with summation and, thus, it is local. This implies that the core of {\sc Matlang} cannot check for the presence of a four-clique in a graph (represented as a Boolean matrix), as this query requires at least four variables to be expressed, and neither it can express the non-local matrix inversion query. It can be shown that {\sc Matlang} is strictly contained in the tame version of \lara\ that is mentioned above, and thus some of our results can be seen as generalizations of the ones for {\sc Matlang}.

\smallskip
\noindent
{\bf Organization of the paper.} Basics of \lara\ and $\foagg$ are presented in Sections \ref{sec:basics} and \ref{se:foagg}, respectively. 
The expressive completeness of \lara\ in terms of $\foagg$ is shown in Section \ref{sec:exp-comp}. The tame version of \lara\ and some inexpressibility 
results relating to it are given in Section \ref{sec:exp-lara}, while in Section \ref{sec:conv} we present a version of \lara\ that can express convolution and some discussion about its expressive power. We finalize in Section \ref{sec:conc} with concluding remarks and future work. 
Due to space constraints some of our proofs are in the appendix. 

\section{The \lara\ Language} \label{sec:basics}
\newcommand{\K}{{\cal K}} 
\newcommand{\V}{{\cal V}} 

\newcommand{\Dom}{{\sf Dom}} 
\newcommand{\sort}{{\sf sort}} 
\newcommand{\keys}{{\sf Keys}} 
\newcommand{\values}{{\sf Values}} 

\newcommand{\solve}{{\sf Solve}}

For integers $m \leq n$, we write $[m,n]$ for
$\{m,\dots,n\}$ and $[n]$ for $\{1,\dots,n\}$. If $\bar v = (v_1,\dots,v_n)$ is a tuple  
of elements, we write ${\bar v}[i]$ for $v_i$. We denote
 multisets as $\{\!\!\{a,b,\dots\}\!\!\}$. 

\paragraph*{Data model} 

A {\em relational schema}  is a finite collection $\sigma$ of {\em two-sorted} relation symbols. 
The first sort consists of {\em key-attributes} and the second one of {\em value-attributes}. 
Each relation symbol $R \in \sigma$ is then associated with a pair $(\bar K,\bar V)$, where $\bar K$ and $\bar V$ are
 (possibly empty) tuples of different key- and value- attributes, respectively.  We write $R[\bar K,\bar V]$ to denote that
  $(\bar K,\bar V)$ is the {\em sort} of $R$. 
We do not distinguish between $\bar K$, resp., $\bar V$, and the set of attributes mentioned in it. 

There are two countably infinite sets of objects over which databases are populated: A domain of
{\em keys}, which interpret key-attributes and is denoted $\keys$, and 
a domain of {\em values}, which interpret value-attributes and is denoted $\values$. 
A {\em tuple of sort $(\bar K,\bar V)$} is a function $t : \bar K \cup \bar V \to \keys \cup \values$ such that 
$t(A) \in \keys$  if $A \in \bar K$ and $t(A) \in \values$  if $A \in \bar V$. A {\em database} $D$ over schema $\sigma$ is a mapping that 
assigns with each relation symbol $R[\bar K,\bar V] \in \sigma$ a finite set $R^D$ of tuples of sort $(\bar K,\bar V)$. We often see 
$D$ as a set of {\em facts}, i.e., as the set of expressions $R(t)$ such that $t \in R^D$. For ease of presentation, we 
write $R(\bar k,\bar v) \in D$ if $R(t) \in D$ for some tuple $t$ with $t(\bar K) = \bar k$ and $t(\bar V) = \bar v$ (where $\bar k \in \keys^{|\bar K|}$ and $\bar v \in \values^{|\bar V|}$).


For a database $D$ to be a {\em \lara\ database} we need $D$ to satisfy an extra restriction: Key attributes define a key 
constraint over the corresponding relation symbols. That is, 
$$R(\bar k,\bar v),R(\bar k,\bar v') \in D \quad \Longrightarrow \quad \bar v = \bar v',$$ 
for each $R[\bar K,\bar V] \in \sigma$, $\bar k \in \keys^{|\bar K|}$, and $\bar v,\bar v' \in \values^{|\bar V|}$.   
Relations of the form $R^D$ are called {\em associative tables} \cite{lara}.  
Yet, we abuse terminology and call associative table to any set $A$ of tuples of the same sort $(\bar K,\bar V)$   such 
that $\bar v = \bar v'$ for each $(\bar k,\bar v), (\bar k,\bar v') \in A$. 
In such a case, $A$ is of sort $(\bar K,\bar V)$. Notice that for a tuple $(\bar k,\bar v)$ in $A$, we can safely
denote $\bar{v}=A(\bar{k})$.

\paragraph*{Syntax} 
An {\em aggregate operator} over domain $U$ is a family $\oplus = \{\oplus_{0},\oplus_{1},\dots,\oplus_{\omega}\}$ 
of partial functions, where each $\oplus_{k}$ takes a multiset of $k$ elements from $U$ and returns a single element in $U$. 
If $u$ is a collection of $k$ elements in $U$, we write $\oplus(u)$ for $\oplus_{k}(u)$. 
This notion generalizes  most aggregate operators 
used in practical query languages; e.g., ${\sf SUM}$, ${\sf AVG}$, ${\sf MIN}$, ${\sf MAX}$, and ${\sf COUNT}$.    
For simplicity, we also see binary operations $\otimes$ on $U$ as aggregate operators $\oplus = \{\oplus_{0},\oplus_{1},\dots,\oplus_{\omega}\}$ 
such that $\oplus_2 = \otimes$ and $\oplus_i$ has an empty domain for each $i \neq 2$.  

The syntax of \lara\ is
parameterized by a set of  
{\em extension functions}. This is a collection $\Omega$ of 
user-defined functions $f$ that map each tuple $t$ of sort $(\bar K,\bar V)$ to a finite associative table of sort 
$(\bar K',\bar V')$, for $\bar K \cap \bar K' = \emptyset$ and $\bar V \cap \bar V' = \emptyset$. 
We say that $f$ is of sort $(\bar K,\bar V) \mapsto (\bar K',\bar V')$. 
As an example, an extension function might take a tuple $t = (k,v_1,v_2)$ of sort $(K,V_1,V_2)$, for $v_1,v_2 \in \mathbb{Q}$, and 
map it to a table of sort $(K',V')$ that contains a single tuple $(k,v)$, where $v$ is the average between $v_1$ and $v_2$. 
%

We inductively define the set of expressions in $\text{\lara}(\Omega)$ over schema $\sigma$ as follows.

\begin{itemize}
    \item {\em Empty associative table.} $\emptyset$ is an expression of sort $(\emptyset,\emptyset)$.
    \item {\em Atomic expressions.} If $R[\bar K,\bar V]$ is in $\sigma$, then $R$ is an expression of sort $(\bar K,\bar V)$.
	\item {\em Join.} If $e_{1}$ and $e_2$ are expressions of sort $(\bar K_{1},\bar V_{1})$ and $(\bar K_{2},\bar V_{2})$, respectively, 
    and $\otimes$ is a binary operator over $\values$, then 
    $e_{1} \bowtie_{\otimes} e_{2}$ is an expression of sort $(\bar K_{1} \cup \bar K_{2}, \bar V_{1}\cup \bar V_{2})$.

    \item {\em Union.}  If $e_{1}$, $e_{2}$ are expressions of sort $(\bar K_{1},\bar V_{1})$ and $(\bar K_{2},\bar V_{2})$, respectively, and $\oplus$ is an aggregate operator over $\values$, 
    then $e_{1}\union_{\oplus}e_{2}$ is an expression of sort $(\bar K_{1} \cap \bar K_{2}, \bar V_{1}\cup \bar V_{2})$. 
       
    \item {\em Extend.} For $e$ an expression of sort $(\bar K,\bar V)$ and $f$ a function in $\Omega$ of sort 
    $(\bar K,\bar V) \mapsto (\bar K',\bar V')$, it is the case that ${\sf Ext}_{f} \, e$ is an expression of sort $(\bar K\cup \bar K',\bar V')$.
\end{itemize}
We write $e[\bar K,\bar V]$ to denote that expression $e$ is of sort $(\bar K,\bar V)$. 

\paragraph*{Semantics}

We assume that for every binary operator $\oplus$ over domain $\values$ 
there is a {\em neutral value} $0_\oplus \in \values$.  
Formally, $\oplus(K) = \oplus(K')$, for every multiset $K$ of elements in $\values$ and 
every extension $K'$ of $K$ with an arbitrary number of occurrences of $0_\oplus$.    
An important notion is \emph{padding}. 
Let $\bar V_1$ and $\bar V_2$ be tuples of value-attributes, and $\bar{v}$ a tuple over $\values$ of sort $\bar V_1$.
Then $\operatorname{pad}_\oplus^{\bar V_2}(\bar{v})$ is a new tuple $\bar v'$ over $\values$ of sort $\bar V_1\cup \bar V_2$
such that for each $V \in \bar V_1 \cup \bar V_2$ we have that $v'(V) = v(V)$, if $V \in \bar V_1$, and $v'(V) = 0_\oplus$, otherwise. 

Consider tuples $\bar{k}_1$ and $\bar k_2$ over key-attributes $\bar K_1$ and $\bar K_2$, respectively.
We say that $\bar{k}_1$ and $\bar{k}_2$ are \emph{compatible}, if $\bar{k}_1(K)=\bar{k}_2(K)$ for every $K \in \bar K_1\cap \bar K_2$.
If $\bar{k}_1$ and $\bar{k}_2$ are compatible, one can define the extended tuple $\bar{k}_1\cup \bar{k}_2$ over key-attributes $\bar K_1\cup \bar K_2$.
Also, given a tuple $\bar{k}$ of sort $\bar{K}$, and a set $\bar{K}'\subseteq \bar{K}$, the restriction $\rest{\bar{k}}{\bar{K}'}$ 
of $\bar{k}$ to attributes $\bar{K}'$ 
is the only tuple of sort $\bar{K}'$ that is compatible with $\bar{k}$.
Finally, given a multiset $T$ of tuples $(\bar k,\bar u)$ of the same sort $(\bar K,\bar V)$ 
we define $\solve_\oplus(T)$ as
\begin{multline*}
\solve_\oplus(T) 
\, := \, \{(\bar k,\bar v) \, \mid \,  \text{there exists }\bar{u}\text{ such that }(\bar{k},\bar{u})\in T \text{ and } \\
\text{$\bar v[i] = \bigoplus_{\bar v' \, : \, (\bar k,\bar v') \in T} \bar 
v'[i]$, for each $i \in [|\bar V|]$} \}.$$
\end{multline*}

\begin{figure} 
    \begin{minipage}{.4\linewidth}
        \centering
        \begin{tabular}{ c c || c c c } 
            \textbf{i} & \textbf{j} & $\mathbf{v}_{1}$ & $\mathbf{v}_{2}$ \\
            \hline
            0 & 0 & 1 & 5\\ 
            0 & 1 & 2 & 6\\ 
            1 & 0 & 3 & 7\\
            1 & 1 & 4 & 8\\
            
        \end{tabular}
    \end{minipage}%
    \vspace{0.2cm} 
    \begin{minipage}{.4\linewidth}
        \centering
        \begin{tabular}{ c c || c c c } 
            \textbf{j} & \textbf{k} & $\mathbf{v}_{2}$ & $\textbf{v}_{3}$\\
            \hline
            0 & 0 & 1 & 1\\ 
            0 & 1 & 1 & 2\\ 
            1 & 0 & 1 & 1\\
            1 & 1 & 2 & 1\\
            
        \end{tabular}
    \end{minipage} 
\caption{Associative tables $A$ and $B$ used for defining the semantics of \lara.}
\label{fig:example} 
\end{figure}

The evaluation of a $\text{\lara}(\Omega)$ expression $e$ over schema $\sigma$ 
on a \lara\ database $D$, denoted $e^D$,  
is inductively defined as follows. 
Since the definitions are not so easy to grasp, we use the associative tables $A$ and $B$ in Figure \ref{fig:example} to construct examples. 
Here, $\textbf{i}$, $\textbf{j}$, and $\textbf{k}$ are key-attributes, while
 $\textbf{v}_1$, $\textbf{v}_2$, and $\textbf{v}_3$ are value-attributes. 
\begin{itemize}
\item {\em Empty associative table.} if $e=\emptyset$ then $e^{D} := \emptyset$. 
\item {\em Atomic expressions.}
If $e = R[\bar K,\bar V]$, for $R \in \sigma$, then $e^{D} := R^D$. 
\item {\em Join.} If $e[\bar K_{1} \cup \bar K_{2}, \bar V_{1}\cup \bar V_{2}] = e_{1}[\bar K_{1},\bar V_{1}] \mathrel{\bowtie_{\otimes}} e_{2}[\bar K_{2},\bar V_{2}]$, then 
\begin{multline*}
e^D \, := \,  \, \{(\bar{k}_1\cup \bar{k}_2, \bar{v}_1 \otimes \bar{v}_2) \mid \bar{k}_1\text{ and }\bar{k}_2\text{ are compatible tuples such that}\\
\bar{v}_1=\operatorname{pad}^{\bar V_2}_\otimes(e_1^D(\bar{k}_1)) \text{ and }\bar{v}_2= \operatorname{pad}^{\bar V_1}_\otimes(e_2^D(\bar{k}_2))\}. 
\end{multline*}
Here, $\bar v_1 \otimes \bar v_2$ is a shortening for $(v^1_1 \otimes v_2^1,\dots,v_n^1 \otimes v_n^2)$ assuming that 
$\bar v_1 = (v^1_1,\dots,v_n^1)$ and $\bar v_2 = (v_2^1,\dots,v_2^n)$. 
For example, the result of $A  \mathrel{\bowtie_{\times}} B$ is shown in Figure \ref{fig:example1}, for $\times$ being the usual product on $\mathbb{N}$ and 
$0_\times = 1$. 
\item {\em Union.} If $e[\bar K_{1} \cap \bar K_{2}, \bar V_{1}\cup \bar V_{2}] = e_{1}[\bar K_{1},\bar V_{1}] \union_{\oplus} e_{2}[\bar K_{2},\bar V_{2}]$, then 
\begin{multline*}
e^{D} \, := \solve_{\oplus}\{\!\{ (\bar{k},\bar{v}) \mid 
\bar{k} = \rest{\bar{k}_1}{{\bar K_1} \cap {\bar K_2}}\text{ and }\bar{v}=\operatorname{pad}^{\bar V_2}_\oplus( e_1^{D}(\bar{k}_1))\text{ for some 
$\bar k_1 \in e_1^D$,} \\  
\text{ or } \bar{k} = \rest{\bar{k}_2}{\bar K_1 \cap \bar K_2}\text{ and }\bar{v}=\operatorname{pad}^{\bar V_1}_\oplus( e_2^{D}(\bar{k}_2))\text{ for some 
$\bar k_2 \in e_2^D$}\}\!\}. 
\end{multline*}
In more intuitive terms, $e^D$ is defined by first projecting over $\bar K_1 \cap \bar K_2$ any tuple in $e_1^D$, resp., in 
$e_2^D$. As the resulting set of tuples might 
no longer be an associative table (because there might be many tuples with the same keys), we have to solve the conflicts by applying 
 the given aggregate operator $\oplus$. This is what $\solve_\oplus$ does.    

For example, the result of $A \union_{+}  B$ is shown in Figure \ref{fig:example1}, for $+$ being the addition on $\mathbb{N}$. 
\item {\em Extend.} If $e[\bar K \cup \bar K',\bar V'] = {\sf Ext}_f \, e_1[\bar K,\bar V]$ and $f$ is of sort $(\bar K,\bar V) \mapsto (\bar K',\bar V')$, then 
\[
e^D \, : = \, \{(\bar{k}\cup \bar{k}', \bar{v}') \mid (\bar{k},\bar{v})\in e_1^D, \text{ and } (\bar{k}',\bar{v}')\in f(\bar k,\bar v)\}. 
\]
Notice that in this case $\bar k \cup \bar k'$ always exists as $\bar K \cap \bar K' = \emptyset$. 

As an example, Figure \ref{fig:example1} shows the results of ${\sf Ext}_g \, A$, where $g$ is a function that does the following: If the key corresponding to attribute $\bf i$ is 0,  then the tuple is associated with the associative table of sort $(\emptyset,{\bf z})$ containing only the tuple $(\emptyset,1)$. Otherwise, 
 the tuple is associated with the empty associative table. 

\end{itemize}

    
    \begin{figure} 
 \begin{minipage}{.3\linewidth}
        \centering
        \captionof*{table}{Table $A \bowtie_{\times} B$}
         \begin{tabular}{ c c c || c c c } 
                \textbf{i} & \textbf{j} & \textbf{k} & $\mathbf{v}_{1}$ & $\mathbf{v}_{2}$ & $\mathbf{v}_{3}$ \\
                \hline
                0 & 0 & 0 & 1 & 5 & 1 \\ 
                0 & 0 & 1 & 1 & 5 & 2 \\ 
                0 & 1 & 0 & 2 & 6 & 1 \\
                0 & 1 & 1 & 2 & 12 & 1 \\
                1 & 0 & 0 & 3 & 7 & 1 \\
                1 & 0 & 1 & 3 & 7 & 2 \\
                1 & 1 & 0 & 4 & 8 & 1 \\
                1 & 1 & 1 & 4 & 16 & 1 \\	       
            \end{tabular}
    \end{minipage} 
     \hspace{0.7cm}  
    \begin{minipage}{.3\linewidth}
        \centering
        \captionof*{table}{Table $A \union_{+} B$}
       \begin{tabular}{ c || c c c} 
                \textbf{j} & $\mathbf{v}_{1}$ & $\mathbf{v}_{2}$ & $\mathbf{v}_{3}$ \\
                \hline
                0 & 4 & 14 & 3 \\ 
                1 & 8 & 17 & 2\\ 
            \end{tabular}
    \end{minipage}%
    \hspace{0.3cm} 
     \begin{minipage}{.3\linewidth}
        \centering
        \captionof*{table}{Table ${\sf Ext}_g \, A$}
     \begin{tabular}{ c c|c } 
                \textbf{i} & \textbf{j} & $\mathbf{z}$ \\
                \hline
                0 & 0  & 1 \\ 
                0 & 1  & 1 \\ 
            \end{tabular}
\end{minipage}%
\caption{The tables $A \bowtie_{\times} B$, $A\union_{+} B$,and ${\sf Ext}_f \, A$.}
\label{fig:example1} 
\end{figure}

Several useful operators, as described below, 
can be derived from the previous ones. 
\begin{itemize} 
\item {\em Map operation.} 
An important particular case of ${\sf Ext}_f$ occurs when $f$ is of sort $(\bar K,\bar V) \mapsto (\emptyset,\bar V')$, i.e., 
when $f$ does not extend the keys in the associative table but only modifies the values. Following \cite{lara}, we write this operation as 
${\sf Map}_f$. 
\item {\em Aggregation.} This corresponds to an aggregation over some of the key-attributes of an associative table. 
Consider a \lara\ expression $e_{1}[\bar K_{1},\bar V_{1}]$, an aggregate operator $\oplus$ over $\values$, and
a $\bar{K}\subseteq \bar{K}_1$, then $e=\ \AggOp{e_1}{\bar{K}}{\oplus}$ is an expression of sort $(\bar K,\bar V_{1})$ such that
$e^{D} \, := \solve_{\oplus}\{\!\{ (\bar{k},\bar{v}) \mid 
\bar{k} = \rest{\bar{k}_1}{\bar K}\text{ and }\bar{v}=e_1^{D}(\bar{k}_1) \}\!\}$.
We note that $\AggOp{e_1}{\bar{K}}{\oplus}$ is equivalent to the expression $e_1\union_\oplus {\sf Ext}_{f}(\emptyset)$, where 
$f$ is the function that associates an empty table of sort $(\bar{K},\emptyset)$ with every possible tuple. 
\item {\em Reduction.} The reduction operator, denoted by $\RedOp{}{}{}$, is just a syntactic variation of aggregation
defined as $\RedOp{e_1}{\bar{L}}{\oplus}\equiv\ \AggOp{e_1}{\bar{K}\setminus\bar{L}}{\oplus}$.
\end{itemize} 

Next we provide an example that applies several of these operators.

    \begin{example}
Consider the schema $\rm{Seqs}[(\it{time},\it{batch},\it{features}),(\it{val})]$, 
which represents a typical tensor obtained as the output of a recurrent neural network that processes 
input sequences.
The structure stores a set of \textit{features} obtained when processing input symbols from a sequence,
one symbol at a \textit{time}. For efficiency the network can simultaneously process a \textit{batch} of examples and provide a single
tensor as output.

Assume that, in order to make a prediction one wants to first obtain, for every example, 
the maximum value of every feature over the time steps, and then apply a \emph{softmax} function.
One can specify all this process in \lara\  as follows. 
\begin{eqnarray}
\rm{Max} & = & \RedOp{\rm{Seqs}}{(\textit{time})}{\max(\cdot)} \label{eq:reduction1} \\
\rm{Exp} & = & \map_{\exp(\cdot)} \rm{Max} \label{eq:point-function} \\
\rm{SumExp} & = & \RedOp{\rm{Exp}}{(\textit{features})}{\operatorname{sum}(\cdot)} \label{eq:reduction2} \\
\rm{Softmax} & = & \JoinOp{\rm{Exp}}{\rm{SumExp}}{\div} \label{eq:broadcast}
\end{eqnarray}
Expression~\eqref{eq:reduction1} performs an aggregation over the \textit{time} attribute to obtain the new tensor 
$\mathsf{Max}[(\it{batch},\it{features}),(\it{val})]$ such that $\mathsf{Max}(b,f)=\max_{u=\mathsf{Seqs}(t,b,f)}u$.
That is, $\mathsf{Max}$ stores the maximum value over all time steps (for every feature of every example).
Expression~\eqref{eq:point-function} applies a point-wise exponential function to obtain the tensor
$\mathsf{Exp}[(\it{batch},\it{features}),(\it{val})]$ such that $\mathsf{Exp}(b,f)=\exp(\mathsf{Max}(b,f))$.
In expression~\eqref{eq:reduction2} we apply another aggregation to compute the sum of the exponentials of all the (maximum) features.
Thus we obtain the tensor $\mathsf{SumExp}[({\it batch}),({\it val})]$ 
such that 
\[
\mathsf{SumExp}(b) \, = \, \sum_{f}\mathsf{Exp}(b,f) \, = \, \sum_{f}\exp(\mathsf{Max}(b,f)). 
\]
Finally, expression~\eqref{eq:broadcast} 
applies point-wise 
division over 
the tensors 
$\mathsf{Exp}[(\it{batch},\it{features}),(\it{val})]$ and $\mathsf{SumExp}[({\it batch}),({\it val})]$.
This defines a tensor 
$\mathsf{Softmax}[(\it{batch},\it{features}),(\it{val})]$
such that 
\[
\mathsf{Softmax}(b,f) \, = \, \frac{\mathsf{Exp}(b,f)}{\mathsf{SumExp}(b)} \, = \, \frac{\exp(\mathsf{Max}(b,f))}{\sum_{f'}\exp(\mathsf{Max}(b,f'))}. 
\]
Thus, we effectively compute the \emph{softmax} of the vector of maximum features over time for every example in the batch.
\qed
\end{example}

It is easy to see that for each \lara\ expression $e$ and \lara\ database $D$, the result 
$e(D)$ is always an associative table.  
Moreover, although the elements in the evaluation $e(D)$ of an expression $e$ over 
 $D$ are not necessarily in $D$ (due to the applications of the operator $\solve_\oplus$ and the extension functions in $\Omega$), 
%
all \lara\ expressions are  {\em safe}, i.e., 
$|e^D|$ is finite. 

\begin{proposition} \label{prop:safe} 
Let $e$ be a $\text{\lara}(\Omega)$ expression. Then $e^D$ is a finite associative table, 
for every \lara\ database $D$.   
\end{proposition} 


\section{First-order Logic with Aggregation} \label{se:foagg}
\newcommand{\X}{{\cal X}} 
\newcommand\enc{\text{enc}}

We consider a two-sorted version of FO with aggregation. We thus assume the existence of 
two disjoint and countably infinite sets of {\em key-variables} and {\em value-variables}. The former are denoted $x,y,z,\dots$ and the latter 
$i,j,k,\dots$. In order to cope with the demands of the extension functions  used by \lara\ (as explained later), 
we allow the language to be parameterized by a collection $\Psi$ of user-defined relations $R$ of some sort $(\bar K,\bar V)$. 
%
%
For each $R \in \Psi$ we blur the distinction between the symbol 
$R$ and its interpretation over $\keys^{|\bar K|} \times \values^{|\bar V|}$. 


\paragraph*{Syntax and semantics} 

The language contains terms of two sorts.
\begin{itemize} 
\item {\em Key-terms}: Composed exclusively by the key-variables $x,y,z\dots$. 
\item {\em Value-terms}: Composed by the 
constants of the form $0_\oplus$, for each aggregate operator $\oplus$,   
the value-variables $i,j,\dots$, and the {\em aggregation terms} defined next. 
Let $\tau(\bar x,\bar y,\bar i,\bar j)$ be a value-term 
mentioning only key-variables among those in $(\bar x,\bar y)$ and value-variables among those in $(\bar i,\bar j)$,  
and $\phi(\bar x,\bar y,\bar i,\bar j)$ a formula whose {\em free} key- and value-variables are those in $(\bar x,\bar y)$ and $(\bar i,\bar j)$, respectively (i.e., 
the variables that do not appear under the scope of a quantifier).  Then for each aggregate operator $\oplus$ we have that 
\begin{equation} 
\label{eq:agg-term} 
\tau'(\bar x,\bar i) \, \, := \, \, {\sf Agg}_\oplus \bar y,\bar j \, \big(\tau(\bar x,\bar y,\bar i,\bar j), \phi(\bar x,\bar y,\bar i,\bar j)\big)
\end{equation} 
is a value-term whose free variables are those in $\bar x$ and $\bar i$. 

\end{itemize} 

Let $\Psi$ be a set of relations $R$ as defined above. 
The set of formulas in the language $\foagg(\Psi)$ over schema $\sigma$ is inductively defined as follows:

\begin{itemize} 

\item Atoms $\bot$, $x = y$, and $\iota = \kappa$ are formulas, for $x,y$ key-variables and $\iota,\kappa$ value-terms. 

\item If $R[\bar K,\bar V] \in \sigma \cup \Psi$, then $R(\bar x,\bar \iota)$ is a formula, where 
$\bar x$ is a tuple of key-variables of the same arity as $\bar K$ and $\bar \iota$ is a tuple of value-terms of the same arity as $\bar V$. 

%

\item If $\phi,\psi$ are formulas, then $(\neg \phi)$, $(\phi \vee \psi)$, $(\phi \wedge \psi)$, $\exists x \phi$, and $\exists i \phi$ are formulas, where $x$ and $i$ are key- and value-variables, respectively. 

\end{itemize}  


We now define the semantics of $\foagg(\Psi)$. 
Let  $D$ be a \lara\ database and $\eta$ an {\em assignment} that interprets each key-variable $x$ as an element $\eta(x) \in \keys$ and 
value-variable $i$ as an element $\eta(i) \in \values$. If $\tau(\bar x,\bar i)$ is a value-term only mentioning variables in $(\bar x,\bar i)$, 
we write 
$\tau^D(\eta(\bar x,\bar i))$ for the 
{\em interpretation} of $\tau$ over $D$ when variables are interpreted according to $\eta$. Also, if 
$\phi(\bar x,\bar i)$ 
is a formula of the logic whose free key- and value-variables are those in $(\bar x,\bar i)$,  
we write $D \models \phi(\eta(\bar x,\bar i))$ if $D$ {\em satisfies} $\phi$ when $\bar x,\bar i$ is interpreted 
according to $\eta$, and $\phi^D$ for the set of tuples 
$\eta(\bar x,\bar i)$ such that $D \models \phi(\eta(\bar x,\bar i))$ for some assignment $\eta$..  

The notion of satisfaction is inherited from the semantics of two-sorted 
FO. The notion of interpretation, on the other hand, requires explanation for the case of 
value-terms. Let $\eta$ be an assignment as defined above. 
Constants $0_\oplus$ are interpreted as themselves and value-variables are interpreted over 
$\values$ according to $\eta$.  
 Consider now an aggregate term of the form \eqref{eq:agg-term}. 
Let $D$ be a  \lara\ database and assume that $\eta(\bar x) = \bar k$, for $\bar 
k \in \keys^{|\bar x|}$, and $\eta(\bar i) = \bar v$, for $\bar 
v \in \values^{|\bar i|}$. Let
$(\bar k'_1,\bar v'_1),(\bar k'_2,\bar v'_2),\dots,$ be an enumeration of all tuples $(\bar k',\bar v') \in \keys^{|\bar y|} \times 
\values^{|\bar j|}$ such that 
$D \models \phi(\bar k,\bar k',\bar v,\bar v')$, i.e. there is an assignment $\eta'$ that coincides with $\eta$ over all 
variables in $(\bar x,\bar i)$ and satisfies $\eta'(\bar y,\bar j) = (\bar k',\bar v')$. Then  
%
$$\tau'(\eta(\bar x,\bar i)) \, \, = \, \, \tau'(\bar k,\bar v) \, \, := \, \, \bigoplus \, \{\!\!\{\tau(\bar k,\bar k'_1,\bar v,\bar v'_1),\tau(\bar k,\bar k'_2,\bar v,\bar v'_2),\dots\}\!\!\} \, \in \, \values.$$

\section{Expressive Completeness of \lara\ with respect to $\foagg$} \label{sec:exp-comp}
We prove that $\text{\lara}(\Omega)$ has the same expressive power as a suitable restriction of 
$\foagg(\Psi_\Omega)$, where  
$\Psi_\Omega$ is a set that 
contains relations that represent the extension functions in $\Omega$. 
In particular, for every extension function $f \in \Omega$ 
of sort $(\bar K,\bar V) \mapsto (\bar K',\bar V')$, 
there is a relation $R_f \subseteq \keys^{|\bar K|+|\bar K'|} \times \values^{|\bar V|+|\bar V'|}$ in $\Psi_\Omega$ 
such that for every $(\bar k,\bar v) \in \keys^{|\bar K|} \times \values^{|\bar V|}$: 
$$f(\bar k,\bar v) \, = \, \{(\bar k',\bar v') \, \mid \, (\bar k,\bar k',\bar v,\bar v') \in R_f\}.$$ 

\textcolor{brown}{Since \lara\ is defined in a minimalistic way, we require some assumptions for our expressive completeness 
result to hold. First, we assume that $\keys = \values$, which allows us to interchangeably move from keys 
to values in the language (an operation that \lara\ routinely performs in several of its applications \cite{lara,lara1}). Moreover, we
assume that there are two reserved values, 0 and 1, which are allowed to be used as 
constants in both $\foagg(\Psi_\Omega)$ and \lara, 
but do not appear in any \lara\ database.} This allows us to ``mark'' tuples in some specific
cases, and thus solve an important semantic mismatch between the two languages. In fact, 
\lara\ deals with multisets in their semantics while $\foagg$ is based on sets only. This causes problems, e.g., 
when taking the union of two associative tables $A$ and $B$ both of which contain an occurrence of the same tuple 
$(\bar k,\bar v)$. While \lara\ would treat both occurrences of $(\bar k,\bar v)$ as different in $A \union B$, 
and hence would be forced to restore the ``key-functionality'' of $\bar k$ based on some aggregate operator, for $\foagg$ 
the union of $A$ and $B$ contains only one occurrence of $(\bar k,\bar v)$.    

\paragraph*{From \lara\ to $\foagg$}

We show first that the expressive power of $\text{\lara}(\Omega)$ is bounded by that of 
$\foagg(\Psi_\Omega)$. 

\begin{theorem} \label{theo:lara-2-foagg} 
For every expression $e[\bar K,\bar V]$ of $\text{\lara}(\Omega)$ 
there is a formula $\phi_e(\bar x,\bar i)$ of $\foagg(\Psi_\Omega)$ 
such that $e^D = \phi_e^D$, for every \lara\ database $D$. 
\end{theorem} 

\begin{proof} 
By induction on $e$. The full proof is in the appendix. The basis cases are simple. Instead of the inductive case we present 
an example of how the union operation is translated, as this provides a good illustration of the main ideas behind the proof. 
Assume that we are given expressions $e_1[K_1,K_2,K_3,V_1,V_2]$ and $e_2[K_3,K_2,V_2,V_3]$ of $\text{\lara}(\Omega)$. 
By induction hypothesis, there are formulas $\phi_{e_1}(x_1,x_2,x_3,i_1,i_2)$ and $\phi_{e_2}(x_3',x_2',i_2',i_3)$ of 
$\foagg(\Psi_\Omega)$ 
such that $e_1^D = \phi_{e_1}^D$ and $e_2^D = \phi_{e_2}^D$, for every \lara\ database $D$. 
We want to be able to express $e = e_1 \bowtie_\otimes e_2$ in $\foagg(\Psi_\Omega)$. 

Let us define a formula $\alpha(x,y,z,i,j,k,f)$ as 
\begin{multline*}
\exists i',j',k' \,\bigg(\,\phi_{e_1}(x,y,z,i',j') \, \wedge \, \phi_{e_2}(y,z,j',k') \, \wedge \, \\ 
\big( \, (i = i' \wedge j = j' \wedge k = 0_\otimes \wedge f = 0) \, \vee \, (i = 0_\otimes \wedge j = j' \wedge k = k' \wedge f = 1)\big) \, \bigg).
\end{multline*} 
Notice that when evaluating $\alpha$ over a \lara\ database  $D$ we obtain  
the set of tuples $(k_1,k_2,k_3,v_1,v_2,v_3,f)$ such that there exists tuples of the form 
$(k_1,k_2,k_3,\cdot,\cdot) \in e_1^D$ and $(k_2,k_3,\cdot,\cdot) \in e_2^D$, and 
either one of the following holds: 
\begin{itemize} 
\item 
$e_1^D(k_1,k_2,k_3) = (v_1,v_2)$; $(v_1,v_2,v_3) = \operatorname{pad}_\otimes^{(V_2,V_3)}(v_1,v_2)$; and $f = 0$, or  
\item 
$e_2^D(k_2,k_3) = (v_2,v_3)$; $(v_1,v_2,v_3) = \operatorname{pad}_\otimes^{(V_1,V_2)}(v_2,v_3)$; and $f = 1$. 
\end{itemize} 
%
The reason why we want to distinguish tuples from $e_1^D$ or $e_2^D$ with a 0 or a 1 in the position of variable $f$, respectively, 
it is because it could be the case that 
$\operatorname{pad}_\otimes^{(V_2,V_3)}(v_1,v_2) = \operatorname{pad}_\otimes^{(V_1,V_2)}(v_2,v_3)$.  
%
The semantics of \lara, which is based on aggregation over multisets of tuples, forces us to treat them as two different tuples. 
The  way we do this in $\foagg$ is by distinguishing them with the extra flag $f$. 

We finally define $\phi_e(x,y,z,i,j,k)$ in $\foagg(\Psi_\Omega)$ as 
\begin{multline*} 
\exists i',j',k',f \, \bigg( \alpha(x,y,z,i',j',k',f) \, \wedge \, 
 i = {\sf Agg}_{\otimes} i',j',k',f' \, \big(\,i', \,  \alpha(x,y,z,i',j',k',f)\big) \, \wedge \\ 
\quad \quad \quad \quad \quad \quad \quad \quad j = {\sf Agg}_{\otimes} i',j',k',f' \, \big(\,j', \,  \alpha(x,y,z,i',j',k',f)\big)
\, \wedge \\ 
k = {\sf Agg}_{\otimes} i',j',k',f' \, \big(\,k', \,  \alpha(x,y,z,i',j',k',f)\big) \, \bigg).
\end{multline*}  
That is, the evaluation of $\phi_e$ on $D$ outputs all tuples $(k_1,k_2,k_3,v_1,v_2,v_3)$ such that: 
\begin{enumerate}
\item There are tuples $(k_1,k_2,k_3,w_1,w_2,0_\otimes,\cdot,0)$ and 
$(k_1,k_2,k_3,0_\otimes,w'_2,w'_3,1)$ in $\alpha^D$. 
\item 
$v_1 = w_1 \otimes 0_\otimes = w_1$; $v_2 = w_2 \otimes w'_2$; and 
$v_3 = 0_\otimes \otimes w'_3 = w'_3$.   
\end{enumerate} 
Clearly, then, $\phi_e^D = e^D$ over every \lara\ database  
$D$. 
\end{proof} 

\textcolor{blue}{Notice that the translation from \lara\ to $\foagg$ given in the proof of Theorem 
\ref{theo:lara-2-foagg} does not require the use of negation. In the next section we show that at least safe 
negation can be encoded in \lara\ by a suitable combination of aggregate operators and extension functions.}  

\paragraph*{From $\foagg$ to \lara}

We now prove that the other direction holds under suitable restrictions and 
assumptions on the language. First, we need to impose two restrictions on $\foagg$ formulas, which ensure that 
the semantics of the formulas considered matches that of \lara. In particular, we need to ensure that the evaluation of  $\foagg$ formulas is safe and 
 only outputs associative tables. 
\begin{itemize} 
\item {\em Safety.} Formulas of $\foagg(\Psi_\Omega)$ are not necessarily safe, i.e., their evaluation can have infinitely many tuples 
(think, e.g., of the formula $i = j$, for $i,j$ value-variables, or 
$R_f(\bar x,\bar x',\bar i,\bar i')$, for $R_f \in \Psi_\Omega$). While safety issues relating to 
the expressive completeness of relational algebra with respect to 
first order logic are often resolved by relativizing all operations to the {\em active} domain of databases (i.e., the set of elements mentioned
in relations in databases), such a restriction only makes sense for keys in our context, but not for values. In fact, several useful formulas compute 
a new value for a variable based on some aggregation terms over precomputed data (see, e.g., the translations of the join and union operator of \lara\ 
into $\foagg$ provided in the proof of Theorem \ref{theo:lara-2-foagg}).   

To overcome this issue we develop a suitable syntactic restriction of the logic that can only express safe queries.
This is achieved by ``guarding'' the application of value-term equalities, relations encoding extension functions, 
 and Boolean connectives as follows. 
 \begin{itemize} 
 \item We only allow equality of value-terms to appear in formulas of the form $\phi(\bar x,\bar i) \, \wedge \, j = \tau(\bar x,\bar i)$, where 
 $j$ is a value-variable that does not necessarily appear in $\bar i$ and $\tau$ is an arbitrary value-term whose value 
 only depends on $(\bar x,\bar i)$. This formula  computes the value of the aggregated term $\tau$ over the precomputed evaluation of $\phi$,  
 and then output it as the value of $j$. In the same vein, atomic formulas of the form $R(\bar x,\bar \iota)$ must satisfy that every element in $\bar \iota$ 
 is a value-variable. 
 \item Relations $R_f \in \Psi_\Omega$ can only appear in formulas of the form 
$\phi(\bar x,\bar i) \wedge R_f(\bar x,\bar x',\bar i,\bar i')$, i.e., we only allow to 
compute the set $f(\bar x,\bar i)$ for specific precomputed values of $(\bar x,\bar i)$. 
 \item Also, negation is only allowed in the restricted form 
 $\phi(\bar x,\bar i) \wedge \neg \psi(\bar x,\bar i)$ and disjunction in the form 
 $\phi(\bar x,\bar i) \vee \psi(\bar x,\bar i)$, i.e., when formulas have exactly the same free variables. 
 \end{itemize} 
  We denote the resulting language as $\foagg^{\rm safe}(\Psi_\Omega)$. These restrictions are meaningful, as the translation from $\text{\lara}(\Omega)$ to $\foagg(\Omega_\Psi)$ given 
in the proof of Theorem \ref{theo:lara-2-foagg} always builds a formula in $\foagg^{\rm safe}(\Psi_\Omega)$. 
 

\item {\em Key constraints.} We also need a restriction on the interpretation of $\foagg^{\rm safe}(\Psi_\Omega)$ 
formulas that ensures that the evaluation of any such a formula 
on a \lara\ database is an associative table. 
For doing this, we modify the syntax of $\foagg^{\rm safe}(\Psi_\Omega)$ formulas 
in such a way that 
every formula $\phi$ of $\foagg^{\rm safe}(\Psi_\Omega)$ now comes equipped with an 
aggregate operator $\oplus$ over $\values$.  
The operator $\oplus$ is used to ``solve'' the key violations introduced 
by the evaluation of $\phi$. Thus, formulas in this section should be understood as pairs 
$(\phi,\oplus)$. 
The evaluation of $(\phi,\oplus)$ over 
 a \lara\ database $D$, denoted $\phi^D_\oplus$, is $\solve_\oplus(\phi^D)$. 
This definition is recursive; e.g., a formula in $\phi(\bar x,\bar i)$ in $\foagg^{\rm safe}(\Psi_\Omega)$ 
which is of the form $\alpha(\bar x,\bar i) \wedge \neg \beta(\bar x,\bar i)$ should now be specified as 
$(\phi,\oplus) = (\alpha,\oplus_\alpha) \wedge \neg (\beta,\oplus_\beta)$. The associative table $\phi^D_\oplus$ corresponds then to 
$$\solve_\oplus(\alpha^D_{\oplus_\alpha} \, \setminus \, \beta^D_{\oplus_\beta}).$$ 
\end{itemize} 

We also require some natural assumptions on the extension functions that \lara\ is allowed to use. In particular, we need 
these functions to be able to express traditional relational algebra operations that are not included in the core of \lara; namely, 
copying attributes, selecting rows based on (in)equality, and projecting over value-attributes (the projection over key-attributes, in turn, can 
be expressed with the union operator). 
Formally, we assume that $\Omega$ contains the following families of extension functions.
\begin{itemize}
%
    
    \item ${\sf copy}_{\bar K,\bar K'}$ and ${\sf copy}_{\bar V,\bar V'}$, for $\bar K,\bar K'$ tuples of key-attributes of the same arity 
    and $\bar V,\bar V'$ tuples of value-attributes of the same arity. Function ${\sf copy}_{\bar K,\bar K'}$ takes as input a tuple 
    $t = (\bar k,\bar v)$ of sort $(\bar K_1,\bar V)$, where $\bar K \subseteq \bar K_1$ and $\bar K' \cap \bar K_1 = \emptyset$,  
    and produces a tuple $t' = (\bar k,\bar k',\bar v)$ of sort $(\bar K_1,\bar K',\bar V)$ such that $t'(\bar K') = t(\bar K)$, i.e., 
    ${\sf copy}_{\bar K,\bar K'}$ copies the value of attributes $\bar K$ in the new attributes 
    $\bar K'$. Analogously, we define the function ${\sf copy}_{\bar V,\bar V'}$. 
    
    \item \textcolor{brown}{${\sf copy}_{\bar V,\bar K}$, for $\bar V$ a tuple of value-attributes and $\bar K$ a  tuple 
    of key-attributes. It takes as input a tuple 
    $t = (\bar k,\bar v)$ of sort $(\bar K_1,\bar V_1)$, where $\bar V \subseteq \bar V_1$ and $\bar K \cap \bar K_1 = \emptyset$,  
    and produces a tuple $t' = (\bar k,\bar k',\bar v)$ of sort $(\bar K_1,\bar K,\bar V)$ such that $t'(\bar V) = t'(\bar K)$, i.e., 
    this function copies the values in $\bar V$ as keys in 
    $\bar K$. (Here it is important our assumption that $\keys = \values$). 
    The reason why this is useful is because \lara\ does not allow to aggregate with respect to values (only with respect to keys), while 
    $\foagg(\Psi_\Omega)$ can clearly do this. Analogously, we define ${\sf copy}_{\bar K,\bar V}$, but this time we copy keys in $\bar K$ to values in 
    $\bar V$. }

    \item ${\sf add}_{V,0_\oplus}$, for $V$ an attribute-value and $\oplus$ an aggregate operator. Function ${\sf add}_{V,0_\oplus}$ takes as input a tuple 
    $t = (\bar k,\bar v')$ of sort $(\bar K,\bar V')$, where $V \not\in \bar V'$,   
    and produces a tuple $t' = (\bar k,\bar v',0_\oplus)$ of sort $(\bar K,\bar V',V)$, i.e., ${\sf add}_{V,0_\oplus}$ adds a new value-attribute $V$
    that always takes value $0_\oplus$. Analogously, we define functions ${\sf add}_{V,0}$ and ${\sf add}_{V,1}$.

    
    \item ${\sf eq}_{\bar K,\bar K'}$ and ${\sf eq}_{\bar V,\bar V'}$, for $\bar K,\bar K'$ tuples of key-attributes of the same arity 
    and $\bar V,\bar V'$ tuples of value-attributes of the same arity. The function 
    ${\sf eq}_{\bar K,\bar K'}$ takes as input a tuple 
    $t = (\bar k,\bar v)$ of sort $(\bar K_1,\bar V)$, where $\bar K, \bar K' \subseteq \bar K_1$, 
    and produces as output the tuple $t' = (\bar k,\bar v)$ of sort $(\bar K_1,\bar V)$, if $t(\bar K) = t(\bar K')$, and the empty associative table otherwise.  Hence, this function acts as a filter over an associative table of sort $(\bar K_1,\bar V)$, extending only those tuples 
$t$ such that $t(\bar K) = t(\bar K')$. 
	Analogously, we define the function ${\sf eq}_{\bar V,\bar V'}$.  
    
    \item In the same vein, extension functions ${\sf neq}_{\bar K,\bar K'}$ and ${\sf neq}_{\bar V,\bar V'}$, for $\bar K,\bar K'$ tuples of key-attributes of the same arity 
    and $\bar V,\bar V'$ tuples of value-attributes of the same arity. These are defined exactly as ${\sf eq}_{\bar K,\bar K'}$ and ${\sf eq}_{\bar V,\bar V'}$, 
    only that we now extend only those tuples $t$ such that $t(\bar K) \neq t(\bar K')$ and $t(\bar V) \neq t(\bar V')$, respectively. 
    

	\item The projection $\pi_{\bar V}$, 
	for $\bar V$ a tuple of value-attributes, takes as input a tuple $(\bar k,\bar v')$ of sort $(\bar K,\bar V')$, where 
	$\bar V \subseteq \bar V'$, and outputs the tuple $(\bar k,\bar v)$ of sort $(\bar K,\bar V)$ such that $\bar v = \rest{{\bar v'}}{\bar{V}}$.  
	
	
\end{itemize}


We now establish our result. 

\begin{theorem} \label{theo:foagg-2-lara} 
Let us assume that $\Omega$ contains all extension functions specified above. 
For every pair $(\phi,\oplus)$, where $\phi(\bar x,\bar i)$ is a formula of $\foagg^{\rm safe}(\Psi_\Omega)$  
and $\oplus$ is an aggregate operator on \values,   
there is a $\text{\lara}(\Omega)$ expression $e_{\phi,\oplus}[\bar K,\bar V]$ such that $e_{\phi,\oplus}^D = \phi_\oplus^D = 
{\sf Solve}_\oplus(\phi^D)$, 
for each \lara\ database $D$. 
\end{theorem} 

Due to lack of space we present this proof in the appendix. 

\textcolor{red}{\paragraph*{Discussion}
The results presented in this section imply that \lara\ has the same expressive power than $\foagg$, 
which in turn is tightly related to the expressiveness of SQL \cite{Lib03}. One might wonder then why to use \lara\ instead of SQL. 
While it is difficult to give a definite answer to this question, we would like to note 
that \lara\  is especially tailored to deal with ML objects, such as matrices or tensors, which are naturally modeled as 
associative tables. As the proof of Theorem \ref{theo:foagg-2-lara} suggests, in turn, $\foagg$ 
requires of several cumbersome tricks to maintain the ``key-functionality'' of associative tables.}  


\section{Expressiveness of \lara\ in terms of ML Operators} \label{sec:exp-lara}


We  assume in this section that $\values = \mathbb{Q}$. 
Since extension functions in $\Omega$ can a priori be arbitrary, to understand what \lara\ can express 
we first need to specify which classes of functions are allowed in $\Omega$. In rough terms, this is determined by the 
operations that one can perform when comparing keys and values, respectively. We explain this below. 
\begin{itemize} 
\item Extensions of two-sorted logics with aggregate operators over a {\em numerical} sort $\cal N$
often permit to perform arbitrary numerical comparisons over $\cal N$  (in our case ${\cal N} = \values = \mathbb{Q}$). It has been noted that 
this extends the expressive power of the language, while at the same time preserving some properties of the logic 
that allow to carry out an analysis of its expressiveness based on well-established techniques (see, e.g., \cite{Lib04}). 
\item In some cases in which the expressive power of the language needs to be further extended, 
one can also define a linear order on the non-numerical sort (which in our case is the set $\keys$) 
and then perform suitable arithmetic comparisons in terms of such a linear order. A well-known application of this idea is in the area 
of descriptive complexity \cite{immerman-book}. 
\end{itemize}  

We start in this section by considering  the first possibility only. That is, we allow 
comparing elements of $\values = \mathbb{Q}$ in terms of arbitrary numerical operations. 
Elements of $\keys$, in turn, can only be compared with respect to equality. 
This yields a logic that is amenable for theoretical exploration -- 
in particular, in terms of its expressive power -- 
and that at the same time is able to express many extension functions of practical interest (e.g., several of the functions used in examples 
in \cite{lara1,lara}).  

We design a simple logic ${\rm FO}(=,{\sf All})$ for expressing extension functions. Intuitively, the name of this logic states 
that it can only 
compare keys with respect to equality but it can compare values in terms of arbitrary numerical predicates. 
The formulas in the logic are 
standard FO formulas where the only atomic expressions allowed are of the following form: 
\begin{itemize} 
\item $x = y$, for $x,y$ key-variables; 
\item $P(i_1,\dots,i_k)$, for $P \subseteq \mathbb{Q}^{k}$ a numerical relation of arity $k$ and $i_1,\dots,i_k$
value-variables or constants of the form $0_\oplus$.  
\end{itemize} 
The semantics of this logic is standard. In particular, an assignment $\eta$ from value-variables to $\mathbb{Q}$ {\em satisfies}
a formula of the form $P(i_1,\dots,i_k)$, for  $P \subseteq \mathbb{Q}^{k}$, whenever 
$\eta(i_1,\dots,i_k) \in P$. 
 
Let $\phi(\bar x,\bar y,\bar i,\bar j)$ be a formula of ${\rm FO}(=,{\sf All})$. 
For a tuple $t = (\bar k,\bar k',\bar v,\bar v') \in \keys^{|\bar k| + |\bar k'|} \times \values^{|\bar v| + |\bar v'|}$ we abuse terminology and say that 
$\phi(\bar k,\bar k',\bar v,\bar v')$ {\em holds} if $D_t \models \phi(\bar k,\bar k',\bar v,\bar v')$, where $D_t$ is the database composed
exclusively by tuple $t$. 
%
%
In addition, an extension function $f$ of sort $(\bar K,\bar V) \mapsto (\bar K',\bar V')$ is {\em definable} in ${\rm FO}(=,{\sf All})$, if
there is a formula $\phi_f(\bar x,\bar y,\bar i,\bar j)$  of ${\rm FO}(=,{\sf All})$, for  
$|\bar x| = |\bar K|$, $|\bar y| = |\bar K'|$, $|\bar i| = |\bar V|$, and $|\bar j| = |\bar V'|$, such that for every tuple $(\bar k,\bar v)$ of sort 
$(\bar K,\bar V)$ it is the case $$f(\bar k,\bar v) \, \, = \, \, \{(\bar k',\bar v') \, \mid \,  \phi(\bar k,\bar k',\bar v,\bar v')  
\text{ holds}\}.$$    
%
%
This gives rise to the definition of the following class of extension functions: 
$$\Omega_{(=,{\sf All})} \, = \, \{ f \, \mid \, \text{$f$ is an extension function that is definable in ${\rm FO}(=,{\sf All})$} \}.$$
Recall that extension functions only produce finite associative tables by definition, and hence only some formulas in ${\rm FO}(=,{\sf All})$
define extension functions. 

%

The extension functions ${\sf copy}_{\bar K,\bar K'}$, ${\sf copy}_{\bar V,\bar V'}$, ${\sf add}_{V,0_\oplus}$, 
${\sf eq}_{\bar K,\bar K'}$, ${\sf eq}_{\bar V,\bar V'}$, ${\sf neq}_{\bar K,\bar K'}$, ${\sf neq}_{\bar V,\bar V'}$, and $\pi_{\bar V}$, 
as shown in the previous section, are in $\Omega_{(=,{\sf All})}$. 
\textcolor{brown}{In turn, ${\sf copy}_{\bar V,\bar K}$ and ${\sf copy}_{\bar K,\bar V}$ are not, as 
${\rm FO}(=,{\sf All})$ cannot compare 
keys with values.} 
Next we provide more examples. 

\begin{example} \label{ex:agg-f}
We use $i + j = k$ and $ij = k$ as a shorthand notation for the 
ternary numerical predicates of addition and multiplication, respectively. 
                Consider first a function $f$ that takes a tuple $t$ of sort $(K_1,K_2,V)$ and computes a tuple $t'$ of sort $
                (K'_1,K'_2,V')$ such that $t(K_1,K_2) = t'(K'_1,K'_2)$ and 
		$t'(V') = 1 -t(V)$. Then $f$ is definable in ${\rm FO}(=,{\sf All})$ as
		$\phi_f(x,y,x',y',i,j) := \big(\, x = x' \, \wedge \, y = y'  \, \wedge \, i + j = 1\,\big)$.  
		This  function can be used, e.g., to interchange 0s and 1s in a Boolean matrix. 
		
		Consider now a function $g$ that takes a tuple $t$ of sort $(K,V_1,V_2)$ and computes a tuple $t'$ of sort $(K',V')$ such that 
		$t(K) = t'(K')$ and 
		$t'(V)$ is the average between $t(V_1)$ and $t(V_2)$. Then $g$ is definable in ${\rm FO}(=,{\sf All})$ as
		$\phi_g(x,y,i_1,i_2,j) := \big(\, x = y \, \wedge \, \exists i \, (i_1+ i_2 = i \, \wedge \, 2j = i)\,\big)$.  
		 \qed
\end{example} 

As an immediate corollary to Theorem \ref{theo:lara-2-foagg} we obtain the following result, which formalizes the 
 fact   
 that -- in the case when $\values = \mathbb{Q}$ --  for translating $\text{\lara}(\Omega_{(=,{\sf All})})$ expressions it is not 
necessary to extend the expressive power of $\foagg$ with the relations in $\Psi_{\Omega_{(=,{\sf All})}}$ as long as one has access 
to all numerical predicates over $\mathbb{Q}$. Formally, let us denote by $\foagg({\sf All})$ the extension of $\foagg$ 
with all formulas of the form $P(\iota_1,\dots,\iota_k)$, 
for $P \subseteq \mathbb{Q}^{k}$
and $\iota_1,\dots,\iota_k$ value-terms, with the expected semantics. 
Then one can prove the following result.  

\begin{corollary} \label{coro:tame-lara-2-foagg}
For every expression $e[\bar K,\bar V]$ of $\text{\lara}(\Omega_{(=,{\sf All})})$ 
there is a formula $\phi_e(\bar x,\bar i)$ of $\foagg({\sf All})$ such that
 $e^D = \phi_e^D$, for every \lara\ database $D$. 
\end{corollary} 

It is known that queries definable in $\foagg({\sf All})$ satisfy two important properties, namely, {\em genericity} and {\em locality}, 
which allow us to prove that neither convolution of matrices nor matrix inversion can be defined in the language. 
From Corollary \ref{coro:tame-lara-2-foagg} we obtain then that none of these queries is expressible in $\text{\lara}(\Omega_{(=,{\sf All})})$. 
We explain this next. 


\paragraph*{Convolution} 

Let $A$ be an arbitrary matrix and $K$ a square matrix. For simplicity we assume that $K$ is of odd size 
$(2n+1)\times (2n+1)$.
The convolution of $A$ and $K$, denoted by $A * K$, is a matrix of the same size as $A$
whose entries are defined as 
\begin{equation}\label{eq:conv}
     (A * K)_{k \ell} \, = \, \displaystyle \sum_{s = 1}^{2n+1} \sum_{t = 1}^{2n+1} {
     A_{k-n+s,\ell-n+t} \cdot K_{s t}}. 
\end{equation}
Notice that $k-n+s$ and $\ell-n+t$ could be invalid indices for matrix $A$. 
The standard way of dealing with this issue is \emph{zero padding}. This simply assumes
those entries outside $A$ to be $0$.
In the context of the convolution operator, one usually calls $K$ a \emph{kernel}.

We represent $A$ and $K$ over the schema $\sigma = \{{\sf Entry}_A[K_1,K_2,V],{\sf Entry}_K[K_1,K_2,V]\}$. Assume that 
$\keys = \{\mathsf{k}_1,\mathsf{k}_2,\mathsf{k}_3,\ldots\}$ 
and $\values = \mathbb{Q}$.  
If $A$ is a matrix of values in $\mathbb{Q}$ of dimension $m \times p$, 
and $K$ is a matrix of values in $\mathbb{Q}$ of dimensions $(2n+1)\times (2n+1)$ with $m,p,n \geq 1$, 
we represent the pair $(A,K)$ 
as the \lara\ database $D_{A,K}$ over $\sigma$ that contains all facts 
${\sf Entry}_A(\mathsf{k}_i,\mathsf{k}_j,A_{i j})$, for $i \in [m]$, $j \in [p]$, and all 
facts 
${\sf Entry}_K(\mathsf{k}_i,\mathsf{k}_j,K_{i j})$, for $i \in [2n+1]$, $j \in [2n+1]$. 
The query ${\sf Convolution}$ over schema $\sigma$ takes as input a
\lara\ database of the form $D_{A,K}$ and returns as output an associative table of sort $[K_1,K_2,V]$
that contains exactly the tuples $(\mathsf{k}_i,\mathsf{k}_j,(A*K)_{ij})$.
We can then prove the following result. 


\begin{proposition}\label{prop:lara-not-conv}
Convolution is not expressible in $\text{\lara}(\Omega_{(=,{\sf All})})$. 
\end{proposition}

The proof, which is in the appendix, is based on a simple {\em genericity} property for the language that is not preserved by convolution.




\paragraph*{Matrix inverse} 
It has been shown by Brijder et al. \cite{BGBW18} that matrix inversion
 is not expressible in \matlang\ by applying techniques based on locality. 
 The basic idea is that \matlang\ is subsumed by $\foagg(\emptyset) = \foagg$, and the latter logic can only define 
 {\em local} properties. Intuitively, this means that formulas in $\foagg$ can only distinguish up to a {\em fixed-radius} neighborhood from 
 its free variables (see, e.g., \cite{Lib04} for a formal definition). On the other hand, as shown in \cite{BGBW18}, 
 if matrix inversion were expressible in \matlang\ there would also be a $\foagg$ formula that defines the transitive closure of a binary relation (represented by its adjacency Boolean matrix). This is a contradiction as transitive closure is the prime example of a non-local property. 
 We use the same kind of techniques to show that 
matrix inversion is not expressible in $\text{\lara}(\Omega_{(=,{\sf All})})$.  For this, we use the fact that 
$\foagg({\sf All})$ is also local. 

We represent Boolean matrices as databases over the schema $\sigma = \{{\sf Entry}[K_1,K_2,V]\}$. Assume that 
$\keys = \mathbb{N}$ and $\values = \mathbb{Q}$.  
The Boolean matrix $M$ of dimension $n \times m$, for $n,m \geq 1$, is represented  
as the \lara\ database $D_M$ over $\sigma$ that contains all facts ${\sf Entry}(i,j,b_{ij})$, for $i \in [n]$, $j \in [m]$, and 
$b_{ij} \in \{0,1\}$, 
such that $M_{ij} = b_{ij}$. Consider the query ${\sf Inv}$ over schema $\sigma$ that takes as input a
\lara\ database of the form $D_M$ and returns as output the \lara\ database $D_{M^{-1}}$, for $M^{-1}$ the inverse of $M$. 
Then: 

\begin{proposition} \label{prop:lara-not-inverse}
$\text{\lara}(\Omega_{(=,{\sf All})})$ cannot express {\sf Inv} over Boolean matrices. That is, there is no $\text{\lara}(\Omega_{(=,{\sf All})})$ 
expression $e_{\sf Inv}[K_1,K_2,V]$over $\sigma$ such that 
$e_{\sf Inv}(D_M) = {\sf Inv}(D_M)$,   
for every  \lara\ database of the form $D_M$ that represents a Boolean matrix $M$.  
\end{proposition}

\section{Adding Built-in Predicates over Keys} \label{sec:conv}
\newcommand{\mA}{A}
\newcommand{\mB}{B}
\newcommand{\mC}{C}
\newcommand{\kk}{\bar k}
\newcommand{\vv}{\bar v}
\newcommand{\xx}{\bar x}
\newcommand{\yy}{\bar y}
\newcommand{\filter}{\operatorname{\mathsf{Filter}}}
\newcommand{\neig}{\mathrm{neighbors}}
\newcommand{\kernel}{\mathrm{kernel}}
\newcommand{\ext}{\mathsf{Ext}}

In Section \ref{sec:exp-lara} we have seen that there are important linear algebra operations, such as matrix inverse and 
convolution, that $\text{\lara}(\Omega_{(=,{\sf All})})$ cannot express. 
The following result shows, on the other hand, that 
a clean extension of $\text{\lara}(\Omega_{(=,{\sf All})})$ can express matrix convolution.  
This extension corresponds to the language \lara$(\Omega_{(<,{\sf All})})$, i.e., 
the extension of  \lara$(\Omega_{(=,{\sf All})})$ in which we assume the existence of a strict linear-order $<$ on $\keys$ and 
extension functions are definable in the logic ${\rm FO}(<,{\sf All})$ that extends ${\rm FO}(=,{\sf All})$ by allowing 
atomic formulas of the form $x < y$, for $x,y$ key-variables. Even more, the only numerical predicates from ${\sf All}$ 
we need are $+$ and $\times$. We denote the resulting logic as \lara$(\Omega_{(<,\{+,\times\})})$. 


\begin{proposition}\label{prop:lara_conv}
{\sc Convolution} is expressible in \lara$(\Omega_{(<,\{+,\times\})})$.
\end{proposition}

The proof of this results is quite cumbersome, and thus it is relegated to the appendix. 

It is worth remarking that Hutchison et al.~\cite{lara1} showed that for every fixed kernel $K$, the query $(A * K)$
is expressible in \lara. However, the \lara\ expression they construct 
depends on the values of $K$, and hence  
their construction does not show that in general convolution is expressible in \lara.
Our construction is stronger, as 
we show that there exists a {\em fixed} \lara$(\Omega_{(<,\{+,\times\})})$ expression that takes 
$A$ and $K$ as input and produces $(A * K)$ as output.

\textcolor{red}{Current ML libraries usually have specific implementations for the convolution operator.
Although specific implementations can lead to very efficient ways of implementing a single convolution, they could prevent the optimization of pipelines that merge several convolutions with other operators.
Proposition~\ref{prop:lara_conv} shows that convolution can be expressed in \lara\ just using
general abstract operators such as aggregation and filtering.
This could open the possibility for optimizing expressions that mix convolution and other operators.}


\paragraph*{Can \lara$(\Omega_{(<,\{+,\times\})})$ express inverse?} 

We believe that \lara$(\Omega_{(<,\{+,\times\})})$ cannot express {\sc Inv}. However, this seems quite challenging to prove. 
First, the tool we used for showing that {\sc Inv} is not expressible in  \lara$(\Omega_{(=,\text{{\sf All}}})$, namely, locality, 
is no longer valid in this setting. In fact, queries expressible in \lara$(\Omega_{(<,\{+,\times\})})$ are not necessarily local. 

\begin{proposition}\label{prop:non_local}
\lara$(\Omega_{(<,\{+,\times\})})$ can express non-local queries.
\end{proposition}


This implies that one would have to 
apply techniques more specifically tailored for the logic, such as {\em Ehrenfeucht-Fra\"iss\'e} games, to show that 
{\sc Inv} is not expressible in \lara$(\Omega_{(<,\{+,\times\})})$. Unfortunately, it is often combinatorially difficult to apply such techniques 
in the presence of built-in predicates, e.g., a linear order, on the domain; cf., \cite{FSV95,Sch96,GS98}. So far, we have not managed to succeed in this regard. 

On the other hand, we can show that {\sc Inv} is not expressible in a natural restriction of 
\lara$(\Omega_{(<,\{+,\times\})})$ under complexity-theoretic assumptions. 
To start with, {\sc Inv} is
complete for the complexity class {\sc Det}, which contains all those problems 
that are logspace reducible to computing the {\em determinant} of a matrix. It is known that {\sc Logspace} $\subseteq$ {\sc Det}, where 
{\sc Logspace} is the class of functions computable in logarithmic space, and this inclusion is believed to be proper 
\cite{Coo85}. 

In turn, most of the aggregate operators used in practical applications, 
including standard ones such as ${\sf SUM}$, ${\sf AVG}$, ${\sf MIN}$, ${\sf MAX}$, and ${\sf COUNT}$, 
can be computed in {\sc Logspace} (see, e.g., \cite{Con90}). 
Combining this with well-known results 
on the complexity of computing relational algebra and arithmetic operations, we obtain that 
the fragment $\text{\lara}_{\rm st}(\Omega_{(<,\{+,\times\})})$ of \lara$(\Omega_{(<,\{+,\times\})})$ that only mentions the standard 
aggregate operators above, and whose formulas defining 
extension functions are {\em safe}, can be evaluated in {\sc Logspace} in {\em data complexity}, i.e., assuming formulas to be fixed. (The definition of safety 
and the proof of the result are given in the appendix).  

\begin{proposition} \label{prop:logspace} 
Let $e[\bar K,\bar V]$ be a fixed expression of $\text{\lara}_{\rm st}(\Omega_{(<,\{+,\times\})})$. There is a {\sc Logspace} procedure that takes as input a \lara\ database $D$ and 
computes $e^D$. 
\end{proposition} 

Hence, proving {\sc Inv} to be expressible in the language 
$\text{\lara}_{\rm st}(\Omega_{(<,\{+,\times\})})$ would imply the surprising result that {\sc Logspace} $=$ {\sc Det}.


\section{Final Remarks and Future Work} \label{sec:conc}
We believe that the work on query languages for analytics systems that integrate relational and 
statistical functionalities provides interesting perspectives for database theory. In this paper we focused on the \lara\ 
language, which has been designed  to become the core algebraic language for such systems, and carried out a systematic study of its expressive power 
in terms of logics and concepts traditionally studied in the database theory literature.   

As we have observed, expressing interesting ML operators in \lara\ requires the addition of complex features, such as arithmetic predicates on the numerical 
sort and built-in predicates on the domain. The presence of such features complicates the study of the expressive power of the languages, as some known techniques no longer hold, e.g., genericity and locality, while others become combinatorially difficult to apply, e.g., 
Ehrenfeucht-Fra\"iss\'e games. In addition, the presence of a built-in linear order might turn the logic capable 
of characterizing some parallel complexity classes, and thus inexpressibility results could be as hard to prove as 
some longstanding conjectures in complexity theory. 

A possible way to overcome these problems might be not looking at languages in its full generality, 
but only at extensions of the tame fragment \lara$(\Omega_{(=,{\sf All})})$ with some of the most sophisticated operators. For instance, 
what if we extend \lara$(\Omega_{(=,{\sf All})})$ directly with an operator that computes {\sf Convolution}? Is it possible to prove that 
the resulting language (\lara$(\Omega_{(=,{\sf All})})$ + {\sf Convolution}) cannot express matrix inverse {\sf Inv}? Somewhat 
a similar approach has been 
followed in the study of {\sc Matlang}; e.g., \cite{BGBW18} studies the language ({\sc Matlang} + {\sc Inv}),  
which extends {\sc Matlang} with the matrix inverse operator. 

Another interesting line of work corresponds to identifying which kind of operations need to be added to \lara\ in order to be able to express in a natural way 
recursive operations such as matrix inverse. One would like to do this in a general yet minimalistic way, as adding too much recursive expressive power 
to the language might render it impractical. It would be important to start then 
by identifying the most important recursive operations one needs to perform 
on associative tables, and then abstract from them the minimal primitives that the language needs to possess for expressing such operations.    



\bibliography{lara}

\newpage

\section*{Appendix}
\begin{proof}[Proof of Theorem \ref{theo:lara-2-foagg}]
By induction on $e$. 

\begin{itemize} 

\item If $e = \emptyset$, then $\phi_e = \bot$.  

\item If $e = R[\bar K,\bar V]$, for $R \in \sigma$, then $\phi_e(\bar x,\bar i) = R(\bar x,\bar i)$, where $\bar x$ and $\bar i$ are tuples of distinct key-
and value-variables of the same arity as $\bar K$ and $\bar V$, respectively. 

\item Consider the expression $e[\bar K_{1} \cup \bar K_{2}, \bar V_{1}\cup \bar V_{2}] = e_{1}[\bar K_{1},\bar V_{1}] \mathrel{\bowtie_{\otimes}} e_{2}[\bar K_{2},\bar V_{2}]$, and assume that $\phi_{e_1}(\bar x_1,\bar i_1)$ and $\phi_{e_2}(\bar x_2,\bar i_2)$ are the formulas obtained for 
$e_{1}[\bar K_{1},\bar V_{1}]$ and $e_{2}[\bar K_{2},\bar V_{2}]$, respectively, by induction hypothesis. 
Let us first define a formula $\alpha_e(\bar x_1,\bar x_2,\bar j,f)$ as
\begin{multline*} 
 \exists \bar i_1, \bar i_2 \, \bigg(\, \phi_{e_1}(\bar x_1,\bar i_1) \, \wedge\, \phi_{e_2}(\bar x_2,\bar i_2) \, \wedge \, 
 \chi_{\bar K_1 \cap \bar K_2}(\bar x_1,\bar x_2) \, \wedge \, \\ 
 \big(\, (\bar j = \operatorname{pad}_\otimes^{\bar V_2}(\bar i_1) \, \wedge \, f = 0) \, \vee \,  (\bar j = \operatorname{pad}_\otimes^{\bar V_1}(\bar i_2) \, \wedge \, f = 1) \, \big)\, \bigg), 
\end{multline*} 
assuming that $\bar x_1$ and $\bar x_2$ share no variables; the same holds for $\bar i_1$ and $\bar i_2$;  
 and the formula 
$\chi_{\bar K_1 \cap \bar K_2}(\bar x_1,\bar x_2)$ takes the conjunction of all atomic formulas of the form $y_1 = y_2$, 
where $y_1$ and $y_2$ are variables that appear in $\bar x_1$ and $\bar x_2$, 
respectively, in the position of some attribute $K \in \bar K_1 \cap \bar K_2$. Notice that 
$\bar j = \operatorname{pad}_\otimes^{\bar V_2}(\bar i_1)$ is expressible in $\foagg(\Psi_\Omega)$ as 
a conjunction of formulas of the form $j = k$, for each $j \in \bar j$, where $k$ is the corresponding variable of $\bar i_1$ if $j$ falls in the position 
of an attribute in $\bar V_1$, 
and it is is $0_\otimes$ otherwise. For instance, if $\bar i = (i_1,i_2)$ is of sort $\bar V_1 = (V_1,V_2)$ and $\bar V_2 = (V_2,V_3)$, then 
 $\bar j = \operatorname{pad}_\otimes^{\bar V_2}(\bar i)$ is the formula $j_1 = i_1 \wedge j_2 = i_2 \wedge j_3 = 0_\otimes$.  
Analogously we can define $\bar j = \operatorname{pad}_\otimes^{\bar V_1}(\bar i_2)$. 
As in the example given in the proof sketch of Theorem \ref{theo:lara-2-foagg} in the main body of the paper, we use the value of $f$ to distinguish 
whether a tuple comes from $e_1^D$ or $e_2^D$. 

It is easy to see that for every \lara\ database $D$ we have that $\alpha_e^D$ computes all tuples of the form $(\bar k_1,\bar k_2,\bar v,w)$
such that $\bar k_1 \in (\exists \bar i_1 \phi_{e_1})^D$, $\bar k_2 \in (\exists \bar i_2 \phi_{e_2})^D$, and $\bar k_1$ and $\bar k_2$ are compatible tuples, and either one of the following statements 
hold:
\begin{itemize} 
\item $\bar v_1 = e_1^D(\bar k_1)$, $\bar v = \operatorname{pad}_\otimes^{\bar V_2}(\bar v_1)$, and $w = 0$; or  
\item $\bar v_2 = e_2^D(\bar k_2)$, $\bar v = \operatorname{pad}_\otimes^{\bar V_1}(\bar v_2)$, and $w = 1$. 
\end{itemize} 
Notice then that for every $(\bar k_1,\bar k_2)$ that belongs to the evaluation of $\exists \bar j,f \, \alpha_e(\bar x_1,\bar x_2,\bar j,f)$ over $D$ there are exactly two tuples 
of the form $(\bar k_1,\bar k_2,\bar v,w) \in \alpha_e^D$: the one which satisfies  $\bar v = \operatorname{pad}_\otimes^{\bar V_2}(e_1^D(\bar k_1))$ and $w = 0$, 
and the one that satisfies $\bar v = \operatorname{pad}_\otimes^{\bar V_1}(e_2^D(\bar k_2))$ and $w = 1$. 

It should be clear then that $\phi_e(\bar x_1,\bar x_2,\bar i)$ can be expressed as: 
$$\exists \bar j,f \, \alpha_e(\bar x_1,\bar x_2,\bar j,f) \, \wedge \, \bigwedge_{\ell \in [|\bar j|]} \bar i[\ell] = 
{\sf Agg}_\otimes \bar j,f \, (\bar j[\ell], \alpha_e(\bar x,\bar j,f)).$$
Notice here that the aggregation is always performed on multisets with exactly two elements (by our previous observation). 
Clearly, the evaluation of $\phi_e$ on $D$, for $D$ a \lara\ database,  contains all 
tuples $(\bar k_1 \cup \bar k_2,\bar v_1 \oplus\bar v_2)$ such that 
$\bar k_1$ and $\bar k_2$ are compatible tuples in $e_1^D$ and $e_2^D$, respectively, and  
it is the case that $\bar v_1 = \operatorname{pad}^{\bar V_2}_\oplus(e_1^D(\bar{k}_1))$ and 
$\bar v_2 = \operatorname{pad}^{\bar V_1}_\oplus(e_2^D(\bar{k}_2))$.

\item Consider the expression $e[\bar K_{1} \cap \bar K_{2}, \bar V_{1}\cup \bar V_{2}] = e_{1}[\bar K_{1},\bar V_{1}] \mathrel{\union_{\oplus}} e_{2}[\bar K_{2},\bar V_{2}]$, and assume that $\phi_{e_1}(\bar x_1,\bar i_1)$ and $\phi_{e_2}(\bar x_2,\bar i_2)$ are the formulas obtained for 
$e_{1}[\bar K_{1},\bar V_{1}]$ and $e_{2}[\bar K_{2},\bar V_{2}]$, respectively, by induction hypothesis. 
We start by defining a formula 
$\alpha_{e_1}(\bar x,\bar j,f)$ as 
$$
\exists \bar x_1,\bar i_1 \, \big(\, \phi_{e_1}(\bar x_1,\bar i_1) \, \wedge \, \eta_1^{\bar K_1 \cap \bar K_2}(\bar x,\bar x_1) \, \wedge \,  
\bar j = \operatorname{pad}_\oplus^{\bar V_2}(\bar i_1) \, \wedge \, f = 0 \, \big),$$
where $\eta_1^{\bar K_1 \cap \bar K_2}(\bar x,\bar x_1)$ states that $\bar x$ is the extension of $\bar x_1$ to represent a tuple in $\bar K_1 \cup \bar K_2$. 
Formally, each variable in $\bar x$ that represents a position in $\bar K_1$ receives the same value than 
the variable in the corresponding position of $\bar x$, and each variable representing a position in $\bar K_2 \setminus \bar K_1$ receives value 0. 
 As an example, if $\bar K_1 = (K_1,K_2)$ and $\bar K_2 = (K_1,K_3)$, then 
$\eta_1^{\bar K_1 \cup \bar K_2}(\bar x,\bar x_1)$ for $\bar x_1 = (y_1,y_2)$ and $\bar x = (z_1,z_2,z_3)$ is $z_1 = y_1 \wedge z_2 = y_2 \wedge z_3 = 0$.
Analogously, we define 
$\alpha_{e_2}(\bar x,\bar j,f)$ as 
$$
\exists \bar x_2,\bar i_2 \, \big(\, \phi_{e_2}(\bar x_2,\bar i_2) \, \wedge \, \eta_2^{\bar K_1 \cap \bar K_2}(\bar x,\bar x_2) \, \wedge \,  
\bar j = \operatorname{pad}_\oplus^{\bar V_1}(\bar i_2) \, \wedge \, f = 1 \, \big).$$
As before, we use distinguished constants 0 and 1 as a way to distinguish tuples coming from $e_1^D$ and $e_2^D$, respectively.

%

Let us now define $\alpha(\bar x,\bar j,f) := \alpha_{e_1}(\bar x,\bar j,f) \vee \alpha_{e_2}(\bar x,\bar j,f)$. 
It is not hard to see then that the evaluation of $\alpha_e$ on a \lara\ database $D$ consists precisely of the tuples of the form $(\bar k,\bar v,w)$
such that one of the following statements 
hold:
\begin{itemize} 
\item $\bar k_1 = \rest{\bar k}{\bar K_1}$ for some $\bar k_1 \in (\exists \bar i_1 \phi_{e_1})^D$;  
 $\bar v = \operatorname{pad}_\oplus^{\bar V_2}(e_1^D(\bar k_1))$;  and $w = 0$; or 

\item $\bar k_2 = \rest{\bar k}{\bar K_2}$ for some $\bar k_2 \in (\exists \bar i_2 \phi_{e_2})^D$, 
 $\bar v = \operatorname{pad}_\oplus^{\bar V_1}(e_2^D(\bar k_1))$, and $w = 1$. 
\end{itemize} 
Let us write $\alpha(\bar x,\bar j,f)$ as $\alpha(\bar x',\bar x'',\bar j,f)$ to denote that $\bar x'$ is the subtuple of 
$\bar x$ that corresponds to variables in $\bar K_1 \cap \bar K_2$, while $\bar x''$ contains all other variables in $\bar x$. 
Notice that in the output of our desired formula $\phi_e$ we are only interested in the value that takes the tuple $\bar x'$. 
It should be clear then that $\phi_e(\bar x,\bar i)$ can be expressed as: 
\begin{multline*}
\exists \bar x',\bar x'', \bar j,f \, \big(\,\bar x = \bar x' \, \wedge \, \alpha(\bar x',\bar x'',\bar j,f) \, \wedge \, \\ 
\bigwedge_{\ell \in [|\bar j|]} \bar i[\ell] = 
{\sf Agg}_\oplus \bar x'',\bar j,f \, (\bar j[\ell], \alpha(\bar x',\bar x'',\bar j,f)\,\big).
\end{multline*}


\item Consider the expression $e[\bar K \cup \bar K',\bar V'] = {\sf Ext}_f \, e_1[\bar K,\bar V]$, where $f$ is of sort 
$(\bar K,\bar V) \mapsto (\bar K',\bar V')$, and assume that $\phi_{e_1}(\bar x_1,\bar i_1)$ is the formula obtained for 
$e_{1}[\bar K_{1},\bar V_{1}]$ by induction hypothesis. It is straightforward to see then that we can define  
$$\phi_e(\bar x_1,\bar x,\bar i) \, := \, \exists \bar i_1 \, \big( \, \phi_{e_1}(\bar x_1,\bar i_1) \, \wedge \, R_f(\bar x_1,\bar x,\bar i_1,\bar i) \,\big).$$

\end{itemize} 
This finishes the proof of the theorem. 
\end{proof} 

\begin{proof}[Proof of Theorem \ref{theo:foagg-2-lara}]
When $f$ is one of the distinguished extension functions $f$ defined above, we abuse notation and write simply $f$ instead of 
${\sf Ext}_f$. 
We first define several useful operations and expressions. 
\begin{itemize} 
\item The projection $\pi^\oplus_{\bar K} \, e$ over keys with respect to aggregate operator $\oplus$, defined as 
$\RedOp{e}{\bar{K}}{\oplus}$. Notice that this removes key-, but not value-attributes from $e$, i.e., if $e$ is of sort $[\bar K',\bar V]$ then 
$\pi^\oplus_{\bar K} \, e$ is of sort $[\bar K,\bar V]$. 
\item 
The {\em rename} operator 
$\rho_{\bar K \to \bar K'} \, e$
as $\pi_{\bar K'} \, ({\sf copy}_{\bar K,\bar K'} \, e)$, where 
$\pi$ has no superscript $\oplus$ as no aggregation is necessary in this case. This operation simply renames the key-attributes $\bar K$ to a fresh set of key-attributes $\bar K'$. 
Analogously, we define $\rho_{\bar V \to \bar V'} \, e$, $\rho_{\bar V \to \bar K} e$ and $\rho_{\bar K \to \bar V} e$. 
\item The {\em active domain} expression $e_{\sf ActDom}$, which takes as input a \lara\ database $D$ and returns all elements $k \in \keys$
that appear in some fact of $D$. It is defined as follows.  
First choose a key attribute not present in any table of $D$; say it is $Z$.
For each $R[\bar K,\bar V] \in \sigma$ we define an expression $R^{\keys} := \pi_{\emptyset} \, R$, which removes all attribute-values in $\bar V$ from $R$. 
For each $K \in \bar K$ we then define $R^{\keys}_K := \pi_K \, R^{\keys}$ as the set of keys that appear in the position of attribute $K$ in $R[\bar K,\bar V]$ (no need to specify superscript $\oplus$ on $\pi$ in this case). 
Finally, we define $e^R_{\sf ActDom} := \union_{K \in \bar K} \rho_{K \to Z} R^{\keys}_K$ and
$e_{\sf ActDom} := \union_{R \in \sigma} e^R_{\sf ActDom}$. 
\end{itemize} 

%
 
We now prove the theorem by induction on $\phi$. 
\begin{itemize} 
\item If $\phi = \bot$ then $e_{\phi,\oplus} = \emptyset$ for every aggregate operator $\oplus$.
 
\item If $\phi = (x = y)$, for $x,y$ key-variables, then $e_{\phi,\oplus}[K,K'] := {\sf eq}_{K,K'} \big(e_{\sf ActDom}[K] \bowtie \rho_{K \to 
K'} e_{\sf ActDom}[K]\big)$ for every aggregate operator $\oplus$. Notice that there is no need to specify an aggregate operator for $\bowtie$ in this case 
as the tables that participate in the join only consist of keys. 
%
%

\item Consider now $\phi = R(\bar{x},\bar{i})$, for $R \in \sigma$. We assume all variables in $\bar x$ and $\bar i$, respectively, to be pairwise distinct, as repetition of variables can always be simulated with equalities. Then $e_{\phi,\oplus}[\bar K,\bar V] := R[\bar K,\bar V]$ for every aggregate operator 
$\oplus$. 


\item Assume that $(\phi,\oplus) = (\phi',\oplus') \wedge \neg (\phi'',\oplus'')$. 
%
Let $e_{\phi',\oplus'}[\bar K, \bar V]$ and $e_{\phi'',\oplus''}[\bar K, \bar V]$ be the expressions obtained for $(\phi',\oplus')$ and $(\phi'',\oplus'')$, respectively, by induction hypothesis. We construct the expression $e_{\phi,\oplus}$ as follows. 

\begin{itemize} 
\item First, we take the union of $e_{\phi',\oplus'}$ and $e_{\phi'',\oplus''}$, resolving conflicts with an aggregate operator ${\sf func}$ that simply checks 
for which tuples of keys it requires to restore ``key-functionality'' 
after performing the union.  In particular, ${\sf func}$ takes as input a multiset of values. If it
contains more than one element, it returns the distinguished symbol $0$ which appears in no \lara\ database $D$. Otherwise it returns the only element in the multiset. For instance, ${\sf func}(\{\!\!\{a,a\}\!\!\}) = 0$ and ${\sf func}(\{\!\!\{a\}\!\!\}) = a$.  
 
Let us define then $e_1[\bar K,\bar V] := e_{\phi',\oplus'} \union_{\sf func} e_{\phi'',\oplus''}$. Notice that $e_1^D$, for 
$D$ a \lara\ database, contains the tuples $(\bar k,\bar v) \in e_{\phi',\oplus'}^D$ such that there is no tuple of the form $(\bar k,\bar v') \in e_{\phi'',\oplus''}^D$, plus 
the tuples of the form $(\bar k,\bar 0)$ such that there are tuples of the form $(\bar k,\bar v') \in e_{\phi',\oplus'}^D$ and $(\bar k,\bar v'') \in e_{\phi'',\oplus''}^D$. 
In other words, by evaluating $e_1$ on $D$ we have marked with $\bar 0$ those tuples $\bar k$ of keys that are candidates to be removed when 
computing the difference $e_{\phi',\oplus'} \setminus e_{\phi'',\oplus''}$. 

\item  
Second, we take the join $e_{\phi',\oplus'}[\bar K,\bar V] \bowtie e'_1[\bar K,\bar V']$, where $e'_1[\bar K,\bar V'] := \rho_{\bar V \to \bar V'} 
e_1[\bar K,\bar V]$ is obtained by simply renaming $\bar V$ as $\bar V'$ in $e_1$ (there is no need to specify 
an aggregate operator for $\bowtie$ as $\bar V$ and $\bar V'$ have no attributes in common),  
and
apply the extension function 
${\sf eq}_{\bar V,\bar V'}$ over it. It is easy to see that when evaluating the resulting expression $e_\alpha[\bar K,\bar V] := {\sf eq}_{\bar V,\bar V'} (e_{\phi',\oplus'} \bowtie e'_1)$ 
on a \lara\ database $D$, we obtain precisely the tuples $(\bar k,\bar v) \in e_{\phi',\otimes'}^D$ such that there is no tuple of the form $(\bar k,\bar w) \in e_{\phi'',\otimes''}^D$. 
In fact, for those tuples we also have that $(\bar k,\bar v)$ belongs to $e_1$, as explained above, and thus ${\sf eq}_{\bar V,\bar V'}$ applies no filter. 
On the contrary, if there is a tuple of the form $(\bar k,\bar w) \in e_{\phi'',\oplus''}^D$ then $(\bar k,\bar 0) \in e_1^D$ as explained above. This means that 
the filter ${\sf eq}_{\bar V,\bar V'}$ is applied and no tuple of the form $(\bar k,\dots)$ appears in the result. (Notice that 
for the latter to hold we use, in an essential way, the ``key-functionality'' of  tables $e_{\phi',\oplus'}^D$ and $e_{\phi'',\oplus''}^D$ and 
the fact that $0$ does not appear in $D$). 
By definition, then, $e_\alpha^D \subseteq e_{\phi,\oplus}^D$. 

\item Third, we take the join $e_{\phi',\oplus'}[\bar K,\bar V] \bowtie e'_{\phi'',\oplus''}[\bar K,\bar V']$, where $$e'_{\phi'',\oplus''}[\bar K,\bar V'] \, := \, 
\rho_{\bar V \to \bar V'} 
e_{\phi'',\oplus''}[\bar K,\bar V]$$ is obtained by simply renaming $\bar V$ as $\bar V'$ in $e_{\phi'',\oplus''}$ (there is no need to specify 
an aggregate operator for $\bowtie$ as $\bar V$ and $\bar V'$ have no attributes in common),  
and
apply the extension function 
${\sf neq}_{\bar V,\bar V'}$ over it. It is easy to see that when evaluating the resulting expression $e_\beta[\bar K,\bar V] := \pi_{\bar V} {\sf neq}_{\bar V,\bar V'} (e_{\phi',\oplus'} \bowtie e'_{\phi'',\oplus''})$ 
on a \lara\ database $D$, we obtain precisely the tuples $(\bar k,\bar v) \in e_{\phi',\otimes'}^D$ such that there is a tuple of the form $(\bar k,\bar w) \in e_{\phi'',\oplus''}^D$ for which $\bar v \neq \bar w$. This means that the evaluation of 
$e_\beta$ 
on $D$ contains precisely the tuples $(\bar k,\bar v) \in e_{\phi',\oplus'}^D$ that do not belong to $e_\alpha^D$, yet they belong to $e_{\phi',\oplus'}^D \setminus e_{\phi'',\oplus''}^D$. (Notice that 
for the latter to hold we use, in an essential way, the ``key-functionality'' of  tables $e_{\phi',\oplus'}^D$ and $e_{\phi'',\oplus''}^D$).

\item Summing up, we can now define $e_{\phi,\oplus}[\bar K,\bar V]$ as $e_\alpha[\bar K,\bar V] \union e_\beta[\bar K,\bar V]$. Notice that there is no need to specify an aggregate operator for $\union$ here, as by construction we have that there are no tuples $\bar k$ of keys that belong to 
both $\pi_{\bar K} e_\alpha^D$ and $\pi_{\bar K} e_\beta^D$.  
    
\end{itemize} 

\item Assume that $(\phi,\oplus) = (\phi',\oplus') \wedge (\phi'',\oplus'')$. 
%
Let $e_{\phi',\oplus'}[\bar K, \bar V]$ and $e_{\phi'',\oplus''}[\bar K, \bar V]$ be the expressions obtained for $(\phi',\oplus')$ and $(\phi'',\oplus'')$, respectively, by induction hypothesis. We construct the expression $e_{\phi,\oplus}$ by taking the join $e_{\phi',\oplus'}[\bar K,\bar V] \bowtie e'_{\phi'',\oplus''}[\bar K,\bar V']$, where $$e'_{\phi'',\oplus''}[\bar K,\bar V'] \, := \, 
\rho_{\bar V \to \bar V'} 
e_{\phi'',\oplus''}[\bar K,\bar V]$$ is obtained by simply renaming $\bar V$ as $\bar V'$ in $e_{\phi'',\oplus''}$ (there is no need to specify 
an aggregate operator for $\bowtie$ as $\bar V$ and $\bar V'$ have no attributes in common),  
and
apply 
$\pi_{\bar V} {\sf eq}_{\bar V,\bar V'}$ over it. In fact, if a tuple $(\bar k,\bar v)$ is selected by this expression it means that $(\bar k,\bar v) \in e_{\phi',\oplus'}^D 
\cap e_{\phi'',\oplus''}^D$ by definition of ${\sf eq}_{\bar V,\bar V'}$. In turn, if $(\bar k,\bar v)$ is not selected by the expression it means either that 
there is no tuple of the form $(\bar k,\bar w) \in e_{\phi'',\oplus''}^D$, or the unique tuple of the form $(\bar k,\bar w) \in e_{\phi'',\oplus''}^D$ satisfies that 
$\bar v \neq \bar w$ (due to the way in which ${\sf eq}_{\bar V,\bar V'}$ is defined). 
In any of the two cases we have that $(\bar k,\bar v) \not\in e_{\phi',\oplus'}^D 
\cap e_{\phi'',\oplus''}^D$.  


\item Assume that $(\phi,\oplus) = (\phi',\oplus') \vee (\phi'',\oplus'')$. 
%
Let $e_{\phi',\oplus'}[\bar K, \bar V]$ and $e_{\phi'',\oplus''}[\bar K, \bar V]$ be the expressions obtained for $(\phi',\oplus')$ and $(\phi'',\oplus'')$, respectively, by induction hypothesis. We construct the expression $e_{\phi,\oplus}$ by taking the union of the following expressions. 
\begin{itemize} 
\item An expression $e_1$ such that, when evaluated on a \lara\ database $D$, it 
computes the tuples $(\bar k,\bar v) \in e_{\phi',\oplus'}^D$ for which there is no tuple of 
the form $(\bar k,\bar w) \in e_{\phi'',\oplus''}^D$. This can be done in the same way as we constructed 
$e_\alpha$ for the case when $(\phi,\oplus) = (\phi',\oplus') \wedge \neg (\phi'',\oplus'')$ (see above). 
\item Analogously, an expression $e_2$ such that, when evaluated on a \lara\ database $D$, it 
computes the tuples $(\bar k,\bar v) \in e_{\phi'',\oplus''}^D$ for which there is no tuple of 
the form $(\bar k,\bar w) \in e_{\phi',\oplus'}^D$.   
\item An expression $e_3$ such that, when evaluated on a \lara\ database $D$, it 
computes the tuples $(\bar k,\bar v) \in e_{\phi',\oplus'}^D$ for which there is a tuple of 
the form $(\bar k,\bar w) \in e_{\phi'',\oplus''}^D$ that satisfies $\bar v = \bar w$. This can be done 
in the same way as we did in the previous point. 
\item An expression $e_4$ such that, when evaluated on a \lara\ database $D$, it 
computes the tuples $(\bar k,\bar v)$ that are of the form $(\bar k, \bar w_1 \oplus \bar w_2)$ 
for $(\bar k,\bar w_1) \in e_{\phi',\oplus'}^D$ and $(\bar k,\bar w_2) \in e_{\phi'',\oplus''}^D$ with $\bar w_1 \neq \bar w_2$.  
The expression $e_4$ can be defined as $e_\alpha \union e_\beta$, where 
$e_\alpha^D$ contains all tuples $(\bar k,\bar v) \in e_{\phi',\oplus'}^D$ such that there is a tuple of 
the form $(\bar k,\bar w) \in e_{\phi'',\oplus''}^D$ that satisfies $\bar v \neq \bar w$, and 
$e_\beta^D$ contains all tuples $(\bar k,\bar v) \in e_{\phi'',\oplus''}^D$ such that there is a tuple of 
the form $(\bar k,\bar w) \in e_{\phi',\oplus'}^D$ that satisfies $\bar v \neq \bar w$. 
It is easy to see how to express $e_\alpha$ and $e_\beta$ by using techniques similar to the ones 
developed in the previous points. 
\end{itemize} 

\item Assume that $(\phi,\oplus) =  (\phi',\oplus') \wedge k = \tau(\bar{x},\bar{i})$, for a formula $\phi'(\bar x,\bar i)$ and a value-term 
$\tau$ of $\foagg^{\rm safe}(\Psi_\Omega)$, and $k$ a value-variable not necessarily present in $\bar i$. 
We only consider the case when $k$ is not in $\bar i$. The other case is similar. 
Before we proceed we prove the following lemma which is basic for the construction. 

\begin{lemma} \label{lemma:terms}
For every pair $(\alpha,\oplus_\alpha)$, where $\alpha(\bar x,\bar i)$ is a formula of $\foagg^{\rm safe}(\Psi_\Omega)$ and 
$\oplus_\alpha$ is an aggregate operator over \values, and 
for every value-term $\lambda(\bar x,\bar i)$ of $\foagg^{\rm safe}(\Psi_\Omega)$, 
there is an expression $e_{\alpha,\oplus_\alpha,\lambda}[\bar K,\bar V,V_1]$ of 
$\text{\lara}(\Omega)$ such that for every \lara\ database $D$:
$$\text{$(\bar{k},\bar{v},v_1)\in e_{\alpha,\oplus_\alpha,\lambda}^{D} \quad \Longleftrightarrow \quad \big(\,(\bar k,\bar v) \in \alpha_{\oplus_\alpha}^D \,$ and 
$\, \lambda(\bar k,\bar v) = v_1 \, \big)$.}$$ 
\end{lemma} 

\begin{proof} 
We prove this by induction on 
$\lambda$. 
\begin{itemize}

	\item Consider the base case when $\lambda = \ell$, for some variable $\ell \in \bar i$. 
	Then $$e_{\alpha,\oplus_\alpha,\lambda}[\bar K,\bar V,V_1] \, := \, {\sf copy}_{V,V_1} \, (e_{\alpha,\oplus_\alpha}[\bar K,\bar V]),$$ 
	where $V$ is the value-attribute corresponding to 
	variable $\ell$ in $\bar V$ and $V_1$ is a fresh value-attribute.  
	
	\item Consider now the base case when $\lambda = 0_{\otimes}$ for some aggregate operator $\otimes$ (the cases when $\lambda = 0$ and 
	$\lambda = 1$ are analogus). Then
	$$e_{\alpha,\oplus_\alpha,\lambda}[\bar K,\bar V,V_1] \, := \, {\sf add}_{V_1,0_{\otimes}} \, (e_{\alpha,\oplus_\alpha}[\bar K,\bar V]),$$ 
	where $V_1$ is a value-attribute not in $\bar V$. 
	
	\item For the induction hypothesis, assume that $$\lambda(\bar{x},\bar{i}) \, \, = \, \, {\sf Agg}_{\otimes} \bar x', \bar i' \, \big(\lambda'(\bar{x},\bar{x}',\bar{i},\bar{i}'), 
	\alpha'(\bar x,\bar x',\bar i,\bar i')\big),$$ 
	for a formula $(\alpha',\oplus_{\alpha'})$ and a value-term $\lambda'$ of $\foagg^{\rm safe}(\Psi_\Omega)$. 
	In addition, assume that $$e_{\alpha',\oplus_{\alpha'},\lambda'}[\bar K,\bar K',\bar V,\bar V',V_1]$$ is the formula that is obtained for $\alpha'$ and 
	$\lambda'$ by induction hypothesis. 
	Let us define $e_1[\bar K,\bar K']$ as $\pi_\emptyset ({\sf copy}_{\bar V,\bar K'} e_{\alpha,\oplus_\alpha})$, i.e., we simply take $e_{\alpha,\oplus_\alpha}[\bar K,\bar V]$ and 
	copy in $\bar K'$ the values of 
	$\bar V$. We then get rid of the values in $\bar V$ by applying the projection $\pi_\emptyset$.     
	We also define $e_2[\bar K,\bar K'',V_1]$ as 
	$$\pi^\otimes_{\bar K,\bar K''} \, \pi_{V_1} ({\sf copy}_{\bar V,\bar K''} e_{\alpha',\oplus_{\alpha'},\lambda'}[\bar K,\bar K',\bar V,\bar V',V_1]).$$
	That is, in $e_2$ we have all tuples $(\bar k,\bar v,u)$ of sort $(\bar K,\bar K'',V_1)$ such that there exists a tuple 
	of the form $(\bar k,\bar k',\bar v,\bar v',w) \in e_{\alpha',\oplus_{\alpha'},\lambda'}^D$ and $u$ is the aggregate value with respect to $\otimes$ of 
	the multiset of all values $w$ such that a tuple of the form $(\bar k,\bar k',\bar v,\bar v',w)$ belongs to $e_{\alpha',\oplus_{\alpha'},\lambda'}^D$.  
	It should be clear then that we can define $e_{\alpha,\oplus_\alpha,\lambda}[\bar K,\bar V,V_1]$ as 
	$$\pi_{\bar K} \, {\sf copy}_{\bar K',\bar V} \big(\, e_1 \, \bowtie \, (\rho_{\bar K'' \to \bar K'} e_2)\,\big).$$

%
\end{itemize}
This finishes the proof of the lemma. 
\end{proof} 
 

It should be clear then that $e_{\phi,\oplus} = e_{\phi',\oplus',\tau}[\bar K,\bar V,V_1]$, where
$e_{\phi',\oplus',\tau}[\bar K,\bar V,V_1]$ is the expression 
constructed for $(\phi',\oplus')$ and $\tau$ by applying Lemma \ref{lemma:terms}. 


\item Assume that $(\phi,\oplus) =  (\phi',\oplus') \wedge R_f(\bar x,\bar x',\bar i,\bar i')$, where 
$\phi(\bar x,\bar i)$ is a formula of $\foagg^{\rm safe}(\Psi_\Omega)$. Let $e_{\phi',\oplus'}[\bar K,\bar V]$ be the expression obtained for $(\phi',\oplus')$ by induction hypothesis. We can then define the expression $e_{\phi,\oplus}$ 
as $${\sf Ext}_f  (e_{\phi',\oplus'}[\bar K,\bar V]) \, \bowtie \, e_{\phi',\oplus'}[\bar K,\bar V],$$ assuming that $f$ is of sort 
$(\bar K,\bar V) \to (\bar K',\bar V')$. There is no need to specify an aggregate operator for $\bowtie$ here, since by assumption we have that 
$\bar V \cap \bar V' = \emptyset$.

\item The cases $(\phi,\oplus) = \exists x (\phi',\oplus')$ and $(\phi,\oplus) = \exists i (\phi',\oplus')$ can be translated as 
$\pi^\oplus_{\bar K} \, e_{\phi',\oplus'}[K,\bar K,\bar V]$
and $\pi_{\bar V} \, e_{\phi',\oplus'}[\bar K,\bar V,V]$, respectively, assuming that 
$e_{\phi',\oplus'}$ is the expression obtained for $(\phi',\oplus')$ 
by induction hypothesis.

\end{itemize} 
This finishes the proof of the theorem. 
\end{proof}

\begin{figure}
\centering
    ${\sf Entry}_A\;=\;$\begin{tabular}{ c c || c } 
    $K_1$ & $K_2$ &  $V$  \\
\hline
    $\mathsf{k}_1$ & $\mathsf{k}_1$ & 1  \\ 
    $\mathsf{k}_1$ & $\mathsf{k}_2$ & 0  \\
    \vdots &\vdots &\vdots \\
    $\mathsf{k}_2$ & $\mathsf{k}_1$ & 0  \\    
    $\mathsf{k}_2$ & $\mathsf{k}_2$ & 1  \\ 
    $\mathsf{k}_2$ & $\mathsf{k}_3$ & 0  \\    
    \vdots &\vdots &\vdots \\
    $\mathsf{k}_4$ & $\mathsf{k}_3$ & 0  \\ 
    $\mathsf{k}_4$ & $\mathsf{k}_4$ & 1  \\     
    
    \end{tabular}
\quad\quad
    ${\sf Entry}_K\;=\;$\begin{tabular}{ c c || c } 
    $K_1$ & $K_2$ &  $V$  \\
\hline
    $\mathsf{k}_1$ & $\mathsf{k}_1$ & 1  \\ 
    $\mathsf{k}_1$ & $\mathsf{k}_2$ & 1  \\
    \vdots &\vdots &\vdots \\
    $\mathsf{k}_3$ & $\mathsf{k}_2$ & 1  \\ 
    $\mathsf{k}_3$ & $\mathsf{k}_3$ & 1  \\
    \end{tabular}     
\quad\quad
    ${\sf Entry}_{A'}\;=\;$\begin{tabular}{ c c || c } 
    $K_1$ & $K_2$ &  $V$  \\
\hline
    $\mathsf{k}_1$ & $\mathsf{k}_1$ & 1  \\ 
    $\mathsf{k}_1$ & $\mathsf{k}_2$ & 0  \\
    \vdots &\vdots &\vdots \\
    $\mathsf{k}_3$ & $\mathsf{k}_2$ & 0  \\    
    $\mathsf{k}_3$ & $\mathsf{k}_3$ & 1  \\ 
    $\mathsf{k}_3$ & $\mathsf{k}_4$ & 0  \\    
    \vdots &\vdots &\vdots \\
    $\mathsf{k}_4$ & $\mathsf{k}_3$ & 0  \\ 
    $\mathsf{k}_4$ & $\mathsf{k}_4$ & 1  \\     
   
    \end{tabular}
\caption{\lara\ representations for matrices $A$ and $K$ in the proof of Proposition \ref{prop:lara-not-conv}}
\label{fig:lara-mat}
\end{figure}

\begin{proof}[Proof of Proposition \ref{prop:lara-not-conv}]
We first observe that when $\text{\lara}(\Omega_{(=,{\sf All})})$ expressions are interpreted as expressions over matrices,
they are invariant under reordering of rows and columns of those matrices. 
More formally, we make use of \emph{key-permutations} and \emph{key-generic} queries.
A \emph{key-permutation} is an injective function $\pi:\keys\to \keys$.
We extend a key-permutation $\pi$ to be a function over $\keys\cup\values$ by letting $\pi$ be the identity over $\values$.
A formula $\phi(\bar x,\bar i)$ is {\em key-generic} if for every \lara\ database $D$, 
key-permutation $\pi$, and assignment $\nu$, we have that $D\models \phi(\nu(\bar x,\bar i))$ iff 
$\pi(D)\models \phi(\pi(\nu(\bar x,\bar i))$.
The following lemma expresses the self-evident property that formulas in $\foagg({\sf All})$ are key-generic.

\begin{lemma}
\label{lemma:lara-generic}
Every formula $\phi(\bar x,\bar i)$ of $\foagg({\sf All})$ is key-generic. 
\end{lemma} 

With the aid of Lemma \ref{lemma:lara-generic} we can now prove Proposition \ref{prop:lara-not-conv}, as it is easy to show 
that {\sf Convolution} is not key-generic (even when the kernel $K$ is fixed). 
To obtain a contradiction assume that there exists a formula $\varphi(x,y,i)$ in $\foagg({\sf All})$
such that for every \lara\ database $D_{A,K}$ we have that $D_{A,K}\models \varphi(\mathsf{k}_i,\mathsf{k}_j,v)$
iff $(A*K)_{ij}=v$.
Let $A$, $K$, and $A'$ be the following matrices
\[
A = 
\begin{bmatrix}
1 & 0  & 0 & 0\\
0 & 1  & 0  & 0 \\
0 & 0 & 0 & 0\\
0 & 0 & 0 & 1 \\
\end{bmatrix} 
\quad \quad 
K =
\begin{bmatrix}
1 & 1 & 1\\
1 & 1 & 1\\
1 & 1 & 1\\
\end{bmatrix}
\quad \quad 
A' = 
\begin{bmatrix}
1 & 0  & 0 & 0\\
0 & 0  & 0  & 0 \\
0 & 0 & 1 & 0\\
0 & 0 & 0 & 1 \\
\end{bmatrix}
\]
The \lara\ representations for these matrices are depicted in Figure~\ref{fig:lara-mat}.
Consider now the key-permutation $\pi$ such that
$\pi(\mathsf{k}_2)=\mathsf{k}_3$, $\pi(\mathsf{k}_3)=\mathsf{k}_2$, and $\pi$ is the identity for every other value in $\keys$.
It is not difficult to see that $\pi(D_{A,K})=D_{A',K}$.
Now, the convolutions $(A * K)$ and $(A' * K)$ are given by the matrices
\[
(A * K) = 
\begin{bmatrix}
2 & 2  & 1 & 0\\
2 & 2  & 1  & 0 \\
1 & 1 & 2 & 1\\
0 & 0 & 1 & 1 \\
\end{bmatrix}
\quad\quad
(A' * K) = 
\begin{bmatrix}
1 & 1  & 0 & 0\\
1 & 2  & 1  & 1 \\
0 & 1 & 2 & 2\\
0 & 1 & 2 & 2 \\
\end{bmatrix}
\]
We know that $D_{A,K}\models \varphi(\mathsf{k}_1,\mathsf{k}_1,2)$ (since $(A*K)_{11}=2$),
then, since $\varphi$ is generic, we have that $\pi(D_{A,K})\models 
\varphi(\pi(\mathsf{k}_1,\mathsf{k}_1,2))$. 
Thus, since $\pi(D_{A,K})=D_{A',K}$, $\pi(\mathsf{k}_1)=\mathsf{k}_1$, and $\pi$ is the identity over $\values$, 
we obtain that $D_{A',K}\models \varphi(\mathsf{k}_1,\mathsf{k}_1,2)$ which is a contradiction since $(A'*K)_{11}=1\neq 2$.
This proves that {\sf Convolution} is not expressible in $\foagg({\sf All})$.
Hence from Corollary \ref{coro:tame-lara-2-foagg} we obtain that
$\text{\lara}(\Omega_{\tt Agg})$ cannot express {\sf Convolution}.
\end{proof}

 \begin{proof}[Proof of Proposition \ref{prop:lara-not-inverse}]
Assume for the sake of contradiction that $e_{\sf Inv}$ exists. From Theorem \ref{theo:lara-2-foagg} there is a $\foagg^{+,\times}$ formula 
$\phi(x_1,x_2,i)$ that expresses ${\sf Inv}$ over Boolean matrices, i.e., 
$e_{\sf Inv}(D_M) = \phi(D_M) = D_{M^{-1}}$, for every  \lara\ database of the form $D_M$ that represents a Boolean 
matrix $M$. For reasons similar to those observed in \cite[Example 12]{BGBW18}, this implies that there is an $\foagg^{+,\times}$ formula 
$\alpha(x_1,x_2,i)$ over $\sigma$ that expresses the {\em transitive closure} query over the class of 
binary relations represented as Boolean matrices.\footnote{A Boolean matrix $M$ represents binary relation $R$ iff $R = \{(u,v) \, \mid \, M(u,v) = 1\}$.}     
That is, for every  \lara\ database of the form $D_M$ that represents a Boolean 
matrix $M$ of $m \times n$, it is the case $\phi(D_M)$ is the set of tuples $(u,v,b_{uv})$ such that $u \in [m]$, $v \in [n]$, and $b_{uv} \in \{0,1\}$ 
satisfies that $b_{uv} = 1$ iff $(u,v)$ belongs to the transitive closure of the binary relation represented by $M$. 
  
It is well known, on the other hand, that $\foagg^{+,\times}$ can only express local queries (cf., \cite{Lib04}). In particular, this implies that 
there 
is no $\foagg^{+,\times}$
formula $\beta(x,y)$ such that for each finite binary relation $R$ over $\mathbb{N}$ represented as a database 
$D_R = \{{\sf Rel}(u,v) \, \mid \, (u,v) \in R\}$, and each pair $(u,v) \in \mathbb{N} \times \mathbb{N}$, 
it is the case that $\beta(D_R)$ is the set of pairs 
$(u,v)$ in the transitive closure of $R$.  
But this is a contradiction, as from $\alpha(x,y,i)$ we can construct a formula $\beta(x,y)$ of $\foagg^{+,\times}$ 
that satisfies this condition. 
In fact, we can define $\beta(x,y)$ as $\alpha'(x,y,1)$, where $\alpha'$ is obtained by 
\begin{itemize} 
\item first 
replacing each subformula of the form $\exists i' \psi(\bar x,i',\bar i)$ in $\alpha$ with 
$\psi(\bar x,0,\bar i) \vee \psi(\bar x,1,\bar i)$, where 0 is a shorthand for ${\sf Zero}(i') := \neg \exists j (j + i' \neq j)$ and 
1 for ${\sf One}(i') := \neg \exists j (j \cdot i' \neq j)$; and then 
\item 
replacing each atomic formula of the form ${\sf Entry}(x',y',1)$ with ${\sf Rel}(x',y')$, and each of the form  ${\sf Entry}(x',y',0)$ 
with $\neg {\sf Rel}(x',y')$. 
\end{itemize} 
 This finishes the proof of the proposition. 
 \end{proof} 
 
\newcommand{\mK}{K}
\newcommand{\mM}{M}
\newcommand{\mF}{F}

\begin{proof}[Proof of Proposition \ref{prop:lara_conv}]
We organize the proof in three parts. We first show how expressions in \lara$(\Omega_{(<,\{+,\times\})})$ 
can use arbitrary arithmetics over key attributes (this is the tricky part of the proof). We then define some
auxiliary operators, and finally we use everything to easily define the convolution.

\paragraph*{Using arbitrary operators over key- and value-attributes}
Assume for simplicity that $\keys = \mathbb{N}$ and $\values = \mathbb{Q}$. 
Consider first the extension function $f:(K,K',\emptyset)\mapsto(\emptyset,V)$ defined as the following {\rm FO}$(<,\{+,\times\})$ formula: 
\[
\phi_f(x,y,i)\;:=\;\;(x < y) \to \, i = 0\;\;\wedge\;\;\neg(x < y) \to\,  i = 1. 
\]
Now consider a relation ${\sf Entry}_A$ of sort $[K,K',V]$ that represents a square matrix of dimension $n\times n$.
By our definition of $f$ we have that ${\sf Map}_f \, {\sf Entry}_A$ has sort $[K,K',V]$ and its evaluation consists of all triples 
$(x,y,i)\in [n]^3$ such that $i=1$ if $x\geq y$ and $i=0$ otherwise. Consider now the expression
\[
{\sf Ind}_{K,V} \; \; := \;\; \; \AggOp{\map_f \, {\sf Entry}_A}{K}{+}
\]
that aggregates $\map_f \, {\sf Entry}_A$ by summing over $V$ by grouping over $K$. 
The evaluation of ${\sf Ind}_{K,V}$ contains all pairs $(x,i)\in [n]\times[n]$ such that $i$ is the natural number that represents the position of $x$ in the linear order over $\keys$. 
Hence, we have that ${\sf Ind}_{K,V}$ contains all pairs $(x,i)\in [n]\times[n]$ such that $x$ is a key, $i$ is a value, and $x=i$.
This simple fact allows us to express extension functions using arbitrary properties in ${\rm FO}(<,\{+,\times\})$ over key- and value-attributes together, without actually mixing sorts.
For example, consider an associative table $R$ of sort $[K,V_1]$ with $[n]$ as set of keys, and assume that we want to construct a new table $R'$ of sort $[K,V_2]$ such that, for every tuple $(x,v)\in R$, relation $R'$ contains the tuple $(x,j)$ with $j=2x+v$.
We make use of the extension function $g:(K,V,V_1)\mapsto (\emptyset,V_2)$ defined by the formula $\phi_g(x,i,v,j):= (j=2i+v)$.
Notice that $\phi_g$ only mentions variables of the second sort.
We can construct $R'$ as
\begin{eqnarray*}
R':=\map_g(\JoinOp{{\sf Ind}_{K,V}}{R}{+})
\end{eqnarray*}
To see that this works, notice first that ${\sf Ind}_{K,V}$ and $R$ has the same set of keys (the set $[n]$).
Moreover, given that ${\sf Ind}_{K,V}$ and $R$ has no value attribute in common, $\JoinOp{{\sf Ind}_{K,V}}{R}{+}$ is just performing a natural join.
Thus the result of $\JoinOp{{\sf Ind}_{K,V}}{R}{+}$ is a table of sort $[K,V,V_1]$ that contains all tuples $(x,i,v)$ such that $(x,i)\in {\sf Ind}_{K,V}$ and $(x,v)\in R$, or equivalently, all tuples $(x,i,v)$ such that $(x,v)\in R$ and $x=i$.
Then with $\map_g(\JoinOp{{\sf Ind}_{K,V}}{R}{+})$ we generate all tuples $(x,j)$ such that $(x,v)\in R$, $x=i$, and $(j=2i+v)$, or equivalently, all tuples $(x,j)$ such that $(x,v)\in R$, and $(j=2x+v)$, which is what we wanted to obtain.
Given the above discussion, we assume in the following that extension functions are defined by expressions over $\keys$ and $\values$, and then we can write the above expression simply as
\begin{eqnarray*}
R':=\map_{(j=2x+v)}R.
\end{eqnarray*}

\paragraph*{Auxiliary operators}
To easily define the convolution we make use of the \emph{cartesian product}, \emph{filtering} and \emph{renaming}, that we next formalize (they follow the intuitive relational algebra definitions). 
Let $e_1[\bar K_{1},\bar V_{1}]$ and $e_2[\bar K_{2},\bar V_{2}]$ be expressions
such that $\bar K_{1}\cap \bar K_{2}=\bar V_{1}\cap \bar V_{2}=\emptyset$.
The cartesian product is a new expression $e[\bar K_{1},\bar K_{2},\bar V_{1},\bar V_{2}]$
such that for every tuple $(\kk_1,\vv_1)\in e_1^D$ and $(\kk_2,\vv_2)\in e_2^D$ we have that $(\kk_1,\kk_2,\vv_1,\vv_2)\in e^D$ for every database $D$.
It is not difficult to prove that the cartesian product, denoted by $\times$ is expressible in \lara.
Another operator that we need is the $\filter$ operator. 
Given an expression $e_1[K,V]$ and a logical expression $\varphi(\xx,\yy)$, filtering $e_1$ with $\varphi$, 
is a new expression $e_2=\filter_\varphi(e_1)$ that has sort $[K,V]$ and such that for every database $D$ it holds that $(\kk,\vv)$ is in $e_2^D$ if and only if $(\kk,\vv)\in e_1^D$ and $\varphi(\kk,\vv)$ holds.
It is not difficult to prove that the filter operator is also expressible in \lara.
Moreover, by the discussion above, for expressions in \lara$(\Omega_{(<,\{+,\times\})})$ we can use filter expressions as arbitrary ${\rm FO}(<,\{+,\times\})$ formulas over keys and values.
Finally, the renaming operator $\rho_{\alpha}(\mA)$ is a simple operator that just changes the name of attributes of $\mA$ according to the assignment $\alpha$.
Both, filtering and renaming can be defined as a special case of $\ext$.

\newcommand{\midd}{\text{mid}}

\paragraph*{Expressing the convolution}
We now have all the ingredients to express the convolution. 
We first write the convolution definition in a more suitable way so we can easily express the required sums.
Let $\mK$ be a kernel of dimensions $\sf m\times \sf m$ with $\sf m$ an odd number. First define $\sf mid$
as $\frac{\sf m-1}{2}$. Consider now a matrix $A$ of dimension $\sf n_1\times \sf n_2$.
Now for every $\sf (i,j)\in[n_1]\times[n_2]$ one can write the following expression for $(A*K)_{\sf i\sf j}$ 
\begin{multline}
(A*K)_{\sf i\sf j} = \text{sum}\;\; \{\!\{\; A_{\sf s\sf t}\cdot K_{\sf k\sf l}\;\mid \; 
\sf s\in[n_1], t\in[n_2], k,l\in[m]\;\; \text{ and }\;\;
\sf i - mid \leq s \leq i + mid\;\; \text{ and }\;\; \\
\sf j - mid \leq t \leq j + mid \;\; \text{ and }\;\; k=s-i + mid +1 \;\; \text{ and }\;\; l=t-j + mid +1\; \}\!\}
\label{eq:new_convolution}
\end{multline}
Now, in order to implement the above definition we use the following extension functions:
\begin{multline*}
\neig(i,j,s,t,m) := \;\;\;\; \exists \midd\big( 2\times \midd=m-1\;\;\; \wedge \\
i - \midd \leq s \;\;\;\wedge\;\;\; 
s\leq i + \midd \;\;\;\wedge\;\;\;  j -\midd \leq t \;\;\;\wedge\;\;\; t\leq j + \midd\big)
\end{multline*}
\begin{multline*}
\kernel(i,j,k,\ell,s,t,m) := \;\;\;\; \exists \midd\big( 2\times \midd=m-1\;\; \wedge \\
k=s-i + \midd +1 \;\;\; \wedge\;\;\; \ell=t-j + \midd +1	\big)
\end{multline*}
\[
\mathrm{diag}(k,\ell,m):= \;\;\;\; (k=\ell) \to m=1\;\;\; \wedge\;\;\; \neg(k= \ell)\to m=0 
\]
We note that the first two expressions are essentially mimicking the inequalities in (\ref{eq:new_convolution}).
The last expression is intuitively defining the diagonal.

In what follows and for simplicity we will use lowercase letters for key and value attributes in associative tables, so it is simpler to make a correspondence between attributes and variables in the extension functions.
Let ${\sf Entry}_A[(i,j),(v)]$ and ${\sf Entry}_K[(k,\ell),(u)]$ be two associative tables that represents a matrix $A$ and the convolution kernel $K$, respectively. 
We first construct an expression that computes the dimension of the kernel:
\[
\mM = \ \union^{\emptyset}_{\mathrm{+}}\map_{\mathrm{diag}}{\sf Entry}_K.
\]
By the definition of $\mathrm{diag}$ we have that $\mM[\emptyset,(m)]$ has no key attributes 
and a single value attribute $m$ that contains one tuple storing the dimension of $\mK$.
Now we proceed to make the cartesian product of ${\sf Entry}_A$ with itself, ${\sf Entry}_K$ and $\mM$.
For that we need to make a copy of ${\sf Entry}_A$ with renamed attributes.
The cartesian product is 
\[
\mC={\sf Entry}_A\times {\sf Entry}_K\times (\rho_{i\tp s,j\tp t,v\tp w}{\sf Entry}_A)\times \mM.
\]
This produces an associative table of sort $C[(i,j,k,\ell,s,t),(v,u,w,m)]$.
Then we compute the following filters over $C$.
\[
\mF=\filter_{\kernel}(\filter_{\neig}(\mC)).
\]
We note that $F$ has sort $F[(i,j,k,\ell,s,t),(v,u,w,m)]$ (just like $C$).
We also note that for every $(\sf i,\sf j)$ the tuple $(\sf i,\sf j,\sf k,\sf l,\sf s,\sf t)$ is a key in $F$ 
if and only if it satisfies the conditions defining the multyset in Equation~\eqref{eq:new_convolution}. 
Thus to compute what we need, it only remains to multiply and sum, which is done in the following
expression
\[
R=\;\AggOp{(\map_{v^\star = w\cdot u}F)}{ij}{+}.
\]
Thus $R$ is of sort $[(i,j),(v^\star)]$ and is such that $(\sf i,j,v)$ is in $R$ if and only if ${\sf v}=(A*K)_{\sf ij}$.

\end{proof}

\newcommand{\Fam}{F}

\begin{proof}[Proof of Proposition \ref{prop:non_local}]
We show how to use \lara$(\Omega_{(<,\{+,\times\})})$ expressions to mimic the proof of Proposition 8.22 in~\cite{Lib04}.
Assume a schema $\{A[K_1,\emptyset],E[K_1,K_2,\emptyset],P[K_1,\emptyset]\}$ and consider the family $\Fam$ of \lara\ databases defined as follows. $D\in \Fam$ if and only if all the following holds.
\begin{itemize}
\item All keys in $E^D$ and $P^D$ are also in $A^D$.
\item $E^D$ is a disjoint union of a chain and zero or more cycles, that is
\[
E^D=\{(a_0^0,a_1^0),\ldots,(a_{k_0-1}^0,a_{k_0}^0)\}\cup 
\bigcup_{i=1}^K\{(a_0^i,a_1^i),\ldots,(a_{k_i-1}^i,a_{k_i}^i),(a_{k_i}^i,a_0^i)\}
\]
with $K\geq 1$ and all the $a_i^j$s different elements in $\keys$.
\item $P^D$ contains an initial segment of the chain in $E^D$ and may contain some of the cycles in $E^D$, 
that is, there exists a $K< k_0$ and a set $L\subseteq [\ell]$ such that
\[
P^D=\{(a_0^0,a_1^0),\ldots,(a_{K-1}^0,a_{K}^0)\}\cup 
\bigcup_{j\in L}\{(a_0^j,a_1^j),\ldots,(a_{k_j-1}^j,a_{k_i}^j),(a_{k_j}^i,a_0^j)\}
\]
\item $|P^D|\leq \log|A^D|$.
\item For every $a\in P^D$ and $b\in A^D\smallsetminus P^D$ it holds that $a< b$.
\end{itemize}
We prove that there exists a \lara$(\Omega_{(<,\{+,\times\})})$ expression $e[K_1,K_2,V_1,V_2]$ such that for every $D\in \Fam$ it holds that $(a,b)$ is a key $e^D$ if and only if $(a,b)$ is in the transitive closure of $E^D$ restricted to elements in $P^D$.

We will also make use of the following property of {\rm FO}$(<,\{+,\times\})$.
It is known~\cite{Lib04} that there exists an {\rm FO}$(<,\{+,\times\})$ formula $\textit{BIT}(x,y)$ such that $\textit{BIT}(a,b)$ holds if and only if the $b$th bit in the binary expansion of $a$ is $1$.
Similarly as in the proof of Proposition \ref{prop:lara_conv} and by using $\textit{BIT}(x,y)$ 
we can produce an extension function $\text{bit}:[K_1,\emptyset]\mapsto [K_2,\emptyset]$ such that
$\text{bit}(k)$ is a table containing all the keys $i$ such that the $i$th bit of the binary expansion of $k$ is $1$. 
Let ${\sf BIT}[K_1,K_2,\emptyset]$ be defined by the expression ${\sf BIT} = \ext_{\text{bit}}A$.

Now assume that $P^D=\{n_1,n_2,\ldots,n_{\ell}\}$ with $\ell=|P^D|$. By the properties of $D$, we can assume that 
$n_1$ is the minimum value in $<$, and that $n_{i+1}$ is the sucesor of $n_i$ in $<$.
Moreover, we also know that $\ell \leq \log |A^D|$.
All this implies that every set $S$ such that $S\subseteq P^D$ can be represented as an element $\text{code}(S)\in A^D$ as follows.
Let $I_S\subseteq [\ell]$ be such that $S=\{n_i\mid i\in I_S\}$.
Then $\text{code}(S)$ is the element in $A^D$ such that its binary expansion has a $1$ exactly at positions $I_S$, and has a $0$ in every other position.
Thus, to check that an element $x\in A^D$ is in a set $S\subseteq P^D$, it is enough to check that $\textit{BIT}(\text{code}(S),x)$ holds
or similarly, that $(\text{code}(S),x)\in {\sf BIT}^D$.

The above observations allow us to simulate an existential second order quantification over subsets of $P$ as a (first order) quantification over elements in $A$ as follows.
Let $\psi:=\exists S \varphi$ be a second order formula with $\varphi$ an {\rm FO} formula 
that can also mention atoms of the form $S(x)$.
Then we can rewrite $\psi$ as $\psi':=\exists s \varphi'$ such that every atom $S(x)$ in $\varphi$ is replaced by 
$\textit{BIT}(s,x)$.
It is not difficult to use this property and the special form of $E$ in the family $\Fam$ to construct a first order formula $\text{reach}(x,y)$ of the form $\exists s\ \varphi(x,y,s)$ over vocabulary $\{\textit{BIT},E\}$ that is true for a tuple $(a,b)$ if and only if $b$ is reachable from $a$ by following a path in $E$ restricted to the set $P$ (see Proposition 8.22 in~\cite{Lib04}).

Now, from $\text{reach}(x,y)$ one can produce a new {\rm FO} formula $\text{reach}'(x,y)$ over 
$\{{\sf BIT}, E,A\}$, replacing every occurence of $\textit{BIT}(x,y)$ by ${\sf BIT}(x,y)$, 
and adding a conjunction with $A(x)$ for every variable $x$ mentioned, ensuring that the resulting formula is a safe formula.

We note that the obtained formula $\text{reach}'(x,y)$ is a standard {\rm FO} formula (over $\{{\sf BIT}, E,A\}$).
Then by Theorem~\ref{theo:foagg-2-lara} we know that there exists a \lara\ expression 
$e[(K_1,K_2)]$ such that for every database $D'$ over schema $\{{\sf BIT}, E,A\}$ it holds that 
$(x,y)\in e^{D'}$ if and only if ${D'}\models \text{reach}'(x,y)$.
Finally, given that ${\sf BIT}$ can be defined in \lara$(\Omega_{(<,\{+,\times\})})$
we obtain that the transitive closure of $E$ restricted to $P$ can also be defined in \lara$(\Omega_{(<,\{+,\times\})})$.
\end{proof}
 
 \begin{proof}[Proof of Proposition \ref{prop:logspace}] 
 We first explain when an extension function of sort $(\bar K,\bar V) \mapsto (\bar K',\bar V')$ 
 is definable by a {\em safe} formula in ${\rm FO}(<,\{+,\times,\})$. Recall that $f$ is definable in ${\rm FO}(<,\{+,\times\})$ when 
 there is a 
 FO formula $\phi(\bar x,\bar x',\bar i,\bar i')$
that only mentions atomic formulas of the form $x = y$ and $x < y$, for $x,y$ key-variables, and $i+j = k$ and $ij = k$, for $i,j,k$ value-variables or 
constants of the form $0_\oplus$, such that for each tuple $(\bar k,\bar v) \in A$ we have that $f(\bar k,\bar v)$ is precisely the set of tuples 
$(\bar k',\bar v')$ such that $\phi(\bar k,\bar v,\bar k',\bar v')$ holds. 
The safe fragment of this logic is obtained by restricting all negated formulas to be {\em guarded} by $\bar x$ and $\bar i$, i.e., 
forcing them to be of the form $\neg \psi(\bar x,\bar i)$, and all 
formulas of the form 
$x = y$, $x < y$, $i+j = k$, and $ij = k$ to have at most one variable that is not guarded by $\bar x$ and $\bar i$, i.e., 
at most one variable that does not appear in $(\bar x,\bar i)$. It is easy to see that if $\phi$ is safe, then it defines an extension function; this is because
for each tuple $(\bar k,\bar v) $ we have that the set of tuples 
$(\bar k',\bar v')$ such that $\phi(\bar k,\bar v,\bar k',\bar v')$ holds is finite. 
 
We now prove the proposition. Since the expression is fixed we only need to show that each operation used in the expression can be computed in {\sc Logspace}.  
The relational algebra operations of join and union can be computed in {\sc Logspace}; see, e.g., \cite{AHV}. Since, in addition, aggregate operators 
included in $\text{\lara}_{\rm st}(\Omega_{(<,\{+,\times\})})$ can be computed in {\sc Logspace}, we obtain that, given associative tables $A$ and $B$, the results of $A 
\bowtie_\oplus B$ and $A \union_\oplus B$ can be computed in {\sc Logspace}. 

Let us consider now the case of ${\sf Ext}_f \, A$, for $f$ an extension 
function of sort $(\bar K,\bar V) \mapsto (\bar K',\bar V')$ definable in the safe fragment of 
${\rm FO}(<,\{+,\times,\})$. Then $f$ is expressible as a safe 
FO formula $\phi(\bar x,\bar x',\bar i,\bar i')$ that only allow atomic formulas of the form $x = y$ and $x < y$, for $x,y$ key-variables, and $i+j = k$ and $ij = k$, for $i,j,k$ value-variables or 
constants of the form $0_\oplus$. We need to show that for each tuple $(\bar k,\bar v) \in A$ we can compute the set of tuples $(\bar k',\bar v')$ such that 
 $\phi(\bar k,\bar v,\bar k',\bar v')$ holds in {\sc Logspace}. But this is easy to see, due to the safety condition and the fact that relational algebra 
 operations and the arithmetic computations of the form $x = y$, $x < y$, $i + j = k$, and $ij = k$ can be computed in {\sc Logspace} when at most one of the variables in the expression is not guarded. 
 \end{proof}

\end{document}